\documentclass[preprint,prb,aps,floatfix,amsmath,amssymb,superscriptaddress]{revtex4}

\usepackage{graphicx}
\usepackage[all]{xy}

\usepackage{amsmath}
\usepackage{amsfonts}
\usepackage{bm}
\usepackage{color}

\usepackage[pdfpagelabels,pdftex,bookmarks,breaklinks]{hyperref}
\hypersetup{colorlinks, linkcolor=black, citecolor=black}

\usepackage{amsthm}
\newtheorem{lemma}{Lemma}
\newtheorem{definition}{Definition}
\newtheorem{theorem}{Theorem}

\newcommand{\be}{\begin{equation}}
\newcommand{\ee}{\end{equation}}

\DeclareMathOperator{\kr}{ker}
\DeclareMathOperator{\im}{im}
\DeclareMathOperator{\rnk}{rank}

\DeclareMathOperator{\spn}{span}

\newcommand{\hfra}{\rho_{enc}}

\newcommand{\calC}{{\cal C}}
\newcommand{\calS}{{\cal S}}

\newcommand{\calV}{{\cal V}}

\newcommand{\calM}{{\cal M}}
\newcommand{\calQ}{{\cal Q}}
\newcommand{\FF}{\mathbb{F}}

\begin{document}

\title{Homological Product Codes}

\author{S. Bravyi}
\affiliation{IBM T. J. Watson Research Center, Yorktown Heights, NY 10598}

\author{M. B. Hastings}
\affiliation{Microsoft Research, Station Q, CNSI Building, University of California, Santa Barbara, CA, 93106}
\affiliation{Quantum Architectures and Computation Group, Microsoft Research, Redmond, WA 98052}
\begin{abstract}
Quantum codes with low-weight stabilizers known as LDPC codes have been
actively studied recently due to their simple syndrome readout circuits and 
potential applications in fault-tolerant quantum computing. 
However,  all families of quantum LDPC codes known to this date
suffer from a poor distance scaling limited by the square-root of the code length. 
 This is in a sharp contrast with the classical case where good families
of LDPC codes are known that combine constant encoding rate and linear distance. 
Here we propose the first family of good quantum codes with low-weight stabilizers. 
The new codes have  a constant encoding rate, linear distance, and stabilizers acting on at most
$O(\sqrt{n})$ qubits, where $n$ is the code length. 
For comparison, all previously known families of good quantum codes have stabilizers of linear weight. 
Our proof combines  two techniques: randomized constructions of good quantum  codes and the homological product operation  from algebraic topology.  
We conjecture that similar methods can produce good stabilizer codes with stabilizer weight $O(n^\alpha)$ for any $\alpha>0$.
Finally, we apply the homological product to construct new small codes with low-weight stabilizers.
\end{abstract}
\maketitle

\tableofcontents

\newpage

\section{Introduction}
\label{sec:intro}

Classical low density parity check codes are characterized by the property that their parity checks act only on $O(1)$ bits.
Such codes have found numerous applications due to their efficient decoding algorithms based on the belief propagation and high transmission rates approaching the channel capacity limit\cite{Gallager62,MacKay99}. In addition to showing good practical performance, some families of LDPC codes 
are good in the coding theory sense
featuring a linear minimum distance and, at the same time, constant encoding rate. Some LDPC codes are known to 
achieve the Gilbert-Varshamov bound on the code parameters\cite{MacKay99}.

The recently emerged field of quantum error correction attempts to apply coding theory principles to the challenging tasks of fault-tolerant quantum
computing and reliable transmission of quantum states through a noisy communication channel. 
A natural question that we investigate here is  whether good LDPC codes have a quantum counterpart. 
To pose this question formally and motivate it let us highlight main distinctions between classical and quantum error correction. 
Most importantly,  a quantum code must protect encoded states from both bit-flip and phase-flip errors.
Accordingly, the  simplest construction of  quantum codes due to Calderbank, Shor, and Steane~\cite{CSS} (CSS) uses a pair of classical linear codes
$C^Z$ and $C^X$ that are responsible for detecting bit-flip and  phase-flip errors respectively. 
Each basis vector $f$ of $C^Z$ or $C^X$ gives rise to a stabilizer operator which is a product of Pauli operators
$Z$ or $X$ respectively over all qubits in the support of $f$. Valid codewords are quantum states invariant under the action any stabilizer,
whereas corrupted codewords may violate one or several stabilizers. 
The requirement that codewords must satisfy both types of stabilizers simultaneusly
translates to a peculiar condition that the two classical codes must be pairwise  orthogonal, $C^X\subseteq (C^Z)^\perp$.

The second distinction between classical and quantum error correction
applies to the recovery step. Namely,  
one must be able to identify violated stabilizers without 
measuring a state of individual code qubits (which could disturb the encoded state). 
This is usually achieved by measuring only ancillary qubits that
collect the syndrome information. To determine the syndrome of a stabilizer acting on some subset of code qubits $S$, 
the corresponding ancilla has to be coupled to each qubit of $S$ by applying a CNOT gate. Since in practice all gates have
a nonzero error probability and errors introduced by each gate accumulate, fault-tolerance considerations
strongly favor codes in which all stabilizers act only on a few qubits, ideally $O(1)$ qubits~\cite{Gottesman13overhead}. 
Quantum codes used in the
state-of-the-art fault-tolerant schemes such as the surface code family~\cite{Raussendorf07,Fowler09}
are of this type. 

\subsection{Quantum LDPC Codes}

A quantum CSS code encoding $k$ qubits into $n$ qubits
with the minimum distances $d^X,d^Z$ 
is  a pair of classical linear codes $C^X,C^Z\subseteq \{0,1\}^n$ with the following properties:
\begin{enumerate}
\item $C^X\subseteq (C^Z)^\perp$ or, equivalently, $C^Z\subseteq (C^X)^\perp$.
\item $k=n-\dim{(C^X)}-\dim{(C^Z)}$.
\item $d^Z$ is the minimum weight of vectors in  $(C^X)^\perp \backslash C^Z$.
\item $d^X$ is the minimum weight of vectors in $(C^Z)^\perp\backslash C^X$.
\end{enumerate}
Here and below by a weight of a vector or a matrix we mean the number of non-zero entries. 
The distances $d^X$ and $d^Z$ determine the minimum number of single-qubit $X$-type and $Z$-type errors respectively
that can corrupt a codeword without being detected.
We shall use a notation $[[n,k,d]]$ for a CSS code defined above, where
$d=\min{\{d^X,d^Z\}}$ is the worst-case minimum distance.  Let us say that a family of codes is good iff it has a constant encoding rate, $k/n=\Omega(1)$, and
a linear distance, $d=\Omega(n)$.  
A code is LDPC if its stabilizers act only on a few qubits and each qubit is acted upon only by a few stabilizers.
To define this formally we have to assume that the code is specified by a pair of
parity check matrices   $A^Z,A^X$ such that 
$C^Z$ and $C^X$ are the linear spaces spanned by rows of $A^Z$ and $A^X$ respectively.
A CSS code has {\em stabilizer weight} $w$ iff 
\begin{enumerate}
\item[5.] Any row and any column of the parity check matrices $A^Z,A^X$ has weight at most $w$.
\end{enumerate}
A family of codes is called LDPC iff it has constant stabilizer
weight~\footnote{Let us remark that the construction of CSS codes
described in this paper has the nice property that the maximum row and column weights of the parity check 
matrices are always  the same. In general, these are two independent parameters.}  $w=O(1)$.
We will use the notation $[[n,k,d,w]]$ for a CSS code defined above. 
In spite of significant efforts, constructing  good  quantum LDPC codes or merely proving that such codes exist remains an elusive goal.
Here we make a step towards this goal by showing how to combine two previously known techniques:
randomized constructions of good codes and homological constructions of LDPC codes. 

It has been realized early on by Kitaev~\cite{Kitaev2003} that homology theory provides a natural framework
to construct and analyze quantum LDPC codes in a systematic way.
 In this framework, described in detail below, code qubits and parity checks are identified
 with cells of properly chosen dimensions in a cell decomposition of some manifold. 
The toric  code introduced by Kitaev~\cite{Kitaev2003} has parameters
$[[2n,2,\sqrt{n},4]]$ and can be described using homologies of a two-dimensional torus. 
  In spite of being one of the first quantum codes discovered, the toric code 
turned out to be optimal in several respects.
In particular, Aharonov and Eldar showed~\cite{Aharonov11} that any quantum code with $w\le 3$ has bounded distance, $d=O(1)$.
Furthermore, it was shown that $d=O(\sqrt{n})$ for any code with geometrically local stabilizers in the 2D geometry~\cite{Bravyi10tradeoffs,Haah2012logical}. 
Subsequent generalizations of the toric code~\cite{Freedman02systolic,Bombin07,Kim07,Zemor09,Tillich2009,Kovalev12,Michnicki12,Delfosse13,FH13} described
 below improved its encoding rate achieving
$k=\Theta(n)$ and slightly improved the distance achieving  $d=\Theta(\sqrt{n\log{n}})$. 
However, the toric code family is not expected to contain good codes. 

The randomized construction of quantum codes pioneered by Calderbank and Shor~\cite{CSS}
defines a suitable random ensemble of pairwise orthogonal classical codes $C^X,C^Z$ and proves that with high probability the resulting CSS code has linear distance.
In fact, such random CSS codes attain the quantum version of the Gilbert-Varshamov bound,
$k/n=1-2H(d/n)$, where $H(x)$ is the Shannon entropy~\cite{CSS}. An alternative construction
of good codes based on random  encoding circuits with small depth was proposed by Brown and Omar~\cite{Brown13}.
One could expect therefore that good quantum LDPC codes, if exist, are likely to be found using random constructions,
as it was the case for classical LDPC codes~\cite{Gallager62}. 

\subsection{Summary of Results}
The present paper contains two technical contributions.
First, we show how to apply the homology theory framework to construct a random 
ensemble of CSS codes with low weight stabilizers. For a code with $n$ qubits our method produces
stabilizers with weight
$O(\sqrt{n})$.  
In addition, any qubit is acted upon by  at most $O(\sqrt{n})$ stabilizers.
Secondly,  we show that a random code
drawn from this ensemble is good with high probability. This leads to the following result.
\begin{theorem}
\label{thm:main}
For all sufficiently large $n$ there exist a quantum CSS code with parameters
$[[n,c_1n, c_2n,c_3\sqrt{n}]]$,  
where $c_i>0$ are constant coefficients independent of $n$.
\end{theorem}
In contrast, all previous constructions of good quantum codes have stabilizer weight $\Theta(n)$, including randomized constructions.
While this result falls short of proving the existence of good quantum LDPC codes, we believe that it can be improved in several respects;
see the discussion below. 

The key ingredient in the proof of Theorem~\ref{thm:main} is the homological product operation introduced in Ref.~\onlinecite{FH13}.
The homological product takes as input a pair of quantum LDPC codes and produces  a larger LDPC code encoding more qubits
and having larger distance that each of the input codes.
To make this more quantitative,  the homological product of two CSS codes $[[n_a,k_a,d_a,w_a]]$, $a=1,2$, is a CSS code $[[n,k,d,w]]$, where
$n=O(n_1n_2)$, $k=k_1k_2$, $w=w_1+w_2$, and $d\le d_1d_2$.
It should be emphasized that the homological product is different from code concatenation. Although concatenation
of two codes $[[n_a,k_a,d_a]]$ gives a code with parameters $[[n_1n_2,k_1k_2,d_1d_2]]$ which are similar to the ones of the product code,
concatenation does not preserve the property of having low-weight stabilizers.
The homological product is a natural generalization of 
the hypergraph product construction by Tillich and Z\'emor~\cite{Tillich2009}. The latter takes as input a pair of {\em classical} LDPC  codes
$[n,k,d]$ and produces a quantum LDPC code $[[O(n^2),k^2,d]]$. Unfortunately, the hypergraph product cannot
achieve distance growing faster than $O(\sqrt{n})$. 

We construct the desired family of codes by taking the homological product of two random CSS codes. Since random codes are typically not LDPC,
the property of having stabilizer weight $w=w_1+w_2$ in the product code is not really needed in our case. For this reason 
we opted to work with a simplified version of the homological product which we call a ``single sector theory" to distinguish it from 
a ``multiple sector theory" of Ref.~\onlinecite{FH13}. The product code constructed using the single sector theory has
parameters $n=n_1n_2$, $k=k_1k_2$, $d\le d_1d_2$ and stabilizer weight $w\le n_1+n_2$
 (here for simplicity we assume that the input codes have the same
distance $d_a$ for both  $X$-type and $Z$-type errors; see Eq.~(\ref{Kun4}) for the general case).
While the single sector theory does not map LDPC codes to LDPC codes, it has an advantage of being easier to analyze and
requires fewer qubits for the product code. To prove Theorem~\ref{thm:main} we apply the single-sector homological product  
to a pair of random CSS codes with fixed length $n_1=n_2$ and fixed number of logical qubits $k_1=k_2$ such that
$k_a=cn_a$ for some small constant $c$. Since $n=n_1n_2$, this guarantees that the product code has constant encoding rate, $k=k_1k_2=\Omega(n)$
and stabilizer weight $w\le n_1+n_2=O(\sqrt{n})$. Furthermore, since random codes are good with high probability, we have $d_a=\Omega(n_a)$. 
If we assumed optimistically that the product code has distance $d=d_1d_2$ (with high probability), then $d=\Omega(n)$ implying that the product code is good.
Unfortunately, obtaining a lower bound on $d$ in terms of $d_1$ and $d_2$ appears to be a hard problem.
In general it is not true that $d=d_1d_2$, see Section~\ref{sec:small} for counter-examples.

Instead we use statistical arguments and prove that the fraction of input codes leading to the output distance $d<cn$
is less than one for a sufficiently small constant $c$. 
While conceptually this proof  is similar to proving  goodness of random CSS codes
as in Ref.~\onlinecite{CSS}, there are several distinctions. Most notably, 
homological product codes are {\em degenerate} whereas completely random CSS codes are not. 
Recall that a degenerate quantum code has some undetectable errors of weight less than the code distance. 
Such low-weight undetectable errors, obtained  as products of stabilizers, have trivial action on any codeword
and, in the case of homological product codes,  have weight $O(\sqrt{n})$ which is much smaller than the code distance $d=\Omega(n)$.
The proof of Ref.~\onlinecite{CSS} is not applicable to degenerate
codes because it attempts to prove that {\em all} undetectable errors have high weight without 
differentiating between stabilizers and logical operators. 

A natural question is whether the stabilizer weight $w=O(\sqrt{n})$ in Theorem~\ref{thm:main} can be improved 
by considering $m$-fold products. For the single-sector theory, homological
product of $m$ input codes $[[n_a,k_a,d_a]]$ has parameters  $n=\prod_{a=1}^m n_a$, $k=\prod_{a=1}^m k_a$, and 
stabilizer weight $w\le \sum_{a=1}^m n_a$. Suppose all input codes have the same length $n_a=n^{1/m}$ 
and the same number of logical qubits $k_a=cn_a$ for some constant $c$. Then the product code has encoding rate $k/n=c^m$
and stabilizer weight $w\le mn^{1/m}$. Although the distance of the product code is very difficult to compute,
we hope that the statistical arguments developed in this paper can be generalized to the $m$-fold product for $m=O(1)$.
Proving that the product code has distance  $d=\Omega(n)$ would establish existence of good quantum codes
with stabilizer weight $w\le n^\epsilon$ for any constant $\epsilon>0$. Furthermore, in Section~\ref{sec:openproblems},
we propose a proof strategy which, if successful, could reduce the stabilizer weight from $n^\epsilon$ to $O(1)$
at the cost of slightly increasing the code length. 

Since  first quantum devices are likely to involve only a few qubits, a natural question is 
how well the homological product performs for small input codes. In Section~\ref{sec:small} we consider the smallest CSS code
correcting any single-qubit error which is  the Steane $[[7,1,3]]$ code. We show that the product of two Steane codes
gives $[[49,1,9]]$ code with stabilizer weight $w=8$. For comparison, concatenating the Steane code with itself gives
$[[49,1,9]]$ code with stabilizer weight $w=12$.

\subsection{Previous Work}

The observation that the theory of CSS codes has a natural interpretation in terms of homology, in particular $\mathbb{Z}_2$ homology,
goes back to the pioneering works by Kitaev~\cite{Kitaev2003}, Freedman and Meyer~\cite{Freedman01projective}, and
Bombin~\cite{Bombin07}. In this subsection we review some constructions of quantum LDPC codes focusing on those obtained
by homological tools.  We leave aside alternative constructions of LDPC codes based on algebraic and 
graph-theoretic methods~\cite{postol2001proposed,Mackay2004sparse,Camara2005constructions,Aly2008}.

Notable codes include hyperpolic surface codes and color codes~\cite{Zemor09,Delfosse13}  which are
generalizations of the toric code defined on a surface of constant negative curvature and large injectivity 
radius. These codes achieve a constant encoding rate and a slowly growing distance.
The toric code has been generalized to higher-dimensional manifolds by Freedman et al~\cite{Freedman02systolic}.
Using a rather complicated 3D manifold the authors of Ref.~\onlinecite{Freedman02systolic} obtained the first 
(and currently the only) example of a quantum LDPC code with the distance growing faster than $\sqrt{n}$.
This code however has only $O(1)$ logical qubits. In a recent breakthrough work Tillich and Z\'emor~\cite{Tillich2009} proposed
a method of constructing quantum LDPC code from a pair of classical LDPC codes. The hypergraph product codes
of Ref.~\onlinecite{Tillich2009}
were  shown to admit a natural description as a homological product of chain complexes~\cite{FH13}. 
An improved version of the hypergraph product codes was proposed by Kovalev and Pryadko~\cite{Kovalev2012}.
There are also interesting examples of LDPC codes, such as Haah's cubic code~\cite{Haah11}, 
with a large gap between the best known lower and upper bounds on the distance which leaves a possibility of faster
than $\sqrt{n}$ distance scaling.  
We summarize parameters of the known quantum LDPC codes and the new product codes in the table below. 

 \begin{center}
\begin{tabular}{|@{\hskip 5mm}c@{\hskip 5mm}|@{\hskip 5mm}c@{\hskip 5mm}|@{\hskip 5mm}c@{\hskip 5mm}|@{\hskip 5mm}c@{\hskip 5mm}|}
\hline
& $k$ & $d$ & $w$ \\
\hline 
Surface codes & $O(1)$ & $O(\sqrt{n})$ & $4$ \\
\hline
Hyperbolic surface codes & $\Omega(n)$ &  $\Omega(\log{n})$ & $O(1)$ \\
\hline
Generalized 3D toric codes & $O(1)$ & $\Omega(\sqrt{n\log{n}})$ & $O(1)$ \\
\hline
Hypergraph product codes & $\Omega(n)$ & $\Omega(\sqrt{n})$ & $O(1)$ \\
 \hline
Homological product codes (new) & $\Omega(n)$ & $\Omega(n)$ & $O(\sqrt{n})$ \\
\hline
 \end{tabular}
 \end{center} 

We emphasize that our construction produces stabilizer codes, rather than subsystem codes\cite{subsys1,subsys2}.  
The latter can be viewed as regular stabilizer codes in which some subset of logical qubits, known as ``gauge qubits", is not used to encode information. 
 Of particular interest are subsystem LDPC codes~\cite{Bacon06operator,Bombin10subs,Bravyi11subs,Bravyi13subs,Brown12,bacon} in
 which the ``gauge group" generated by stabilizers and logical operators on the gauge qubits has generators of weight $O(1)$.
The recovery step  for a subsystem LDPC code requires only measurements
on subsets of $O(1)$ qubits even if stabilizer generators have a very large weight. 
Indeed, since any stabilizer $S$ belongs to the gauge group, it can be represented as a product of
low-weight gauge group generators, $S=G_1\cdots G_m$. Hence  the syndrome of $S$ can
be determined by measuring eigenvalues of individual generators $G_i$
and classically computing the product of the observed outcomes (here we assume for simplicity
that all generators that appear in the decomposition of $S$ pairwise commute).
However, the above is true only in the idealized settings. 
Since in practice operations performed at the recovery step are noisy, 
 fault-tolerance considerations strongly favor subsystem codes
in which both gauge generators and stabilizer generators have low weight. Indeed, if a stabilizer $S$ as above has too large weight,
the syndrome of $S$ cannot be reliably deduced from {\em noisy} measurements of the gauge generators $G_i$ since 
the measurement errors tend to accumulate.

\subsection{Discussion and Outline}

The homological product can be intuitively understood by considering generalized toric codes as an example.
These codes can be defined on any $D$-dimensional manifold $M$ by applying the following three steps.
 First, one chooses a discretization (for example, a triangulation) of the manifold.  Second, one takes this discretization and constructs a
 chain complex --- a set of vector spaces and certain linear operators on these spaces as reviewed in the next section.  Third, one
 converts the chain complex into a CSS code, 
  as reviewed also in the next section.  Given two manifolds, $M_1,M_2$, a very natural operation is to construct the product manifold $M_1 \times M_2$.  A discretization of the product manifold can be obtained from those of $M_1$ and $M_2$. This gives rise to a new chain complex for $M_1 \times M_2$ and hence a new code.  Crucially for our purposes, the chain complex that corresponds to $M_1 \times M_2$ can be constructed directly from the chain complexes corresponding to $M_1$ and to $M_2$.  This operation
 of constructing a new chain complex from two other chain complexes, is called the homological product.  Since for our purposes a chain complex is equivalent to a CSS quantum code, this allows us to construct a new code from two other codes, in a fashion completely distinct from concatenating the codes.  Rather than applying this homological product
to codes obtained from manifolds with some nice properties, we instead apply it directly to codes obtained from a randomized construction.  

The rest of the paper is organized as follows. 
Section~\ref{sec:review} reviews the construction of codes from homology. Section \ref{sec:rand} constructs a random 
ensemble of good CSS codes. The proof of Theorem~\ref{thm:main} is contained in 
 Section \ref{sec:distance} which gives lower bound on the distance for the homological product of two random codes.
Section \ref{sec:small} presents numerical results on small codes.
Finally, section \ref{sec:openproblems} discusses several open problems.  Appendix~A proves some counting results used in the main text, while Appendix~B  extends the homological construction to $GF(4)$ codes.

\section{Quantum Codes, Homology, and Product Complexes}
\label{sec:review}

In this section we introduce a homological description of CSS codes. 
We first review some standard terminology which may be less familiar to a coding theory audience
and then define  a homological product of two CSS codes which plays the key role in this paper. 
We note that our construction of CSS codes from chain complexes is slightly different from the one
previously described in the literature~\cite{FH13}. Throughout this paper we shall use  notations
$\kr{A}$ and $\im{A}$ for the kernel and the image of a linear map $A$. 

\subsection{Homological Description of  CSS Codes}
\label{subs:CSS}

The theory of CSS codes has a natural interpretation in terms of homology, in particular $\mathbb{Z}_2$ homology.  
Recall that the main object of a homology theory is a {\em chain complex}. It is defined by  
a sequence of spaces, often written ${\cal C}_i$, for certain integers $i$, and by certain linear operators from one space to another. 
In the case of  $\mathbb{Z}_2$ homology, the spaces $\calC_i$ 
 are vector spaces over the binary field $\mathbb{F}_2$ (more generally, they could be vector spaces over other fields or more generally modules).  The linear operators are called boundary operators, and often one writes $\partial_i$ to denote an operator from ${\cal C}_i$ to ${\cal C}_{i-1}$.  The defining requirement of a boundary operator is that
\be
\label{delta^2}
\partial_{i-1} \partial_i = 0.
\ee
This allows us to define a CSS code from a chain complex with three spaces ${\cal C}_2,{\cal C}_1,{\cal C}_0$.  Assign a basis to each of these three spaces. 
 Let there be one qubit per basis vector in ${\cal C}_1$.  Define parity check spaces $C^Z,C^X\subseteq \calC_1$ as
\[
C^Z=\im{\partial_2} \quad \mbox{and} \quad C^X=\im{\partial_1^T}.
\]
Here $\partial_1^T\, : \, \calC_0\to \calC_1$ is obtained by transposing the matrix of $\partial_1$ in the chosen basis. 
To check that $C^X\subseteq (C^Z)^\perp$ choose any vectors $z=\partial_2 f \in C^Z$ and
$x=\partial_1^T g\in C^X$. Here $f\in \calC_2$ and $g\in \calC_0$ are arbitrary vectors.
Then the inner product between $x$ and $z$ is $(x,z)=(g,\partial_1\partial_2 f)=0$ due to Eq.~(\ref{delta^2}).
Thus $C^X,C^Z$ indeed define a CSS code with $n=\dim{\calC_1}$ code qubits. 
This construction can be readily generalized to a construction of CSS codes for qudits rather than qubits, using $\mathbb{Z}_d$ homology instead of $\mathbb{Z}_2$ homology.

In this paper, however, we use a slightly simplified construction which we call a {\em single sector} theory
to distinguish it from the
``multiple sector" theory outlined above. We will see that the single sector theory requires less qubits to build a product of two codes
and is easier to analyze. 
In a single sector theory, a chain complex consists of a single binary linear space ${\cal C}$ and a linear operator $\partial$ mapping ${\cal C}$ to itself.  This linear operator $\partial$ is again called a boundary operator and is required to satisfy the condition 
\be
\label{SST1}
\partial^2=0.
\ee
We will choose $\calC$ as the $n$-dimensional binary space $\mathbb{F}_2^n$ equipped
with the standard basis such that all basis vectors have weight one. Then the transposed matrix
$\partial^T$ is well-defined and $(\partial^T)^2=0$. We define a CSS code 
by choosing the parity check matrices as $A^X=\partial$ and $A^Z=\partial^T$. 
The rows $A^Z$ and $A^X$ span  parity check spaces 
\be
\label{SST2}
C^Z=\im{\partial} \quad \mbox{and} \quad C^X=\im{\partial^T}.
\ee
The condition $\partial^2=0$ implies that 
\begin{equation}
\label{SST3}
(C^Z)^\perp=\kr{\partial^T} \quad \mbox{and} \quad (C^X)^\perp=\kr{\partial}.
\end{equation}
Since $\im{\partial}\subseteq \kr{\partial}$, the parity check spaces are mutually orthogonal,
$C^Z\subseteq (C^X)^\perp$. 
Hence the complex $(\calC,\partial)$ defines a CSS code with $n=\dim{(\calC)}$ code qubits and 
\be
\label{SST4}
k=n-2\rnk{(\partial)}
\ee
logical qubits. The code has stabilizer weight $w$ whenever every row and every column
of $\partial$ has weight at most $w$. Thus LDPC codes correspond to
sparse boundary operators that have $O(1)$ non-zero entries in every row and every column.
 The number of linearly independent parity checks of each type is equal to $\rnk{(\partial)}$.
 One can always get an independent set of parity checks by picking any maximal independent 
subset of columns and rows of $\partial$.
Finally,  the code distances $d^Z$ and $d^X$ coincide with the minimum weight of vectors
in $\kr{\partial}\backslash \im{\partial}$ and $\kr{\partial^T}\backslash \im{\partial^T}$ respectively.

In this paper we adopt some standard terms from homology theory referring
to various elements of the chain complex. For the reader's convenience 
we summarize those terms below and translate them
to the coding theory language.
 \begin{center}
\begin{tabular}{|c|c|c|}
\hline
$(\calC,\partial)$ & complex & CSS code \\
$\partial\, : \, \calC\to \calC$ & boundary operator &  \\
$\kr{\partial}$ & cycles & undetectable errors of $Z$-type \\
$\kr{\partial^T}$ & cocycles & undetectable errors of $X$-type \\
$\im{\partial}$ & trivial cycles & products of $Z$-type stabilizers \\
$\im{\partial^T}$ & trivial cocycles & products of $X$-type stabilizers \\
 $\kr{\partial}\backslash \im{\partial}$ & non-trivial cycles & $Z$-type logical operators  \\
$\kr{\partial^T}\backslash \im{\partial^T}$ & non-trivial cocycles & $X$-type logical operators \\
$\kr{\partial}/\im{\partial}$ & homology class & equivalence class of $Z$-type logical operators  \\
$\kr{\partial^T}/\im{\partial^T}$ & cohomology class & equivalence class of $X$-type logical operators   \\
 \hline
\end{tabular}
\end{center}
The middle column in lines~3-8 of the table 
shows the term for vectors in the particular set defined in the left-hand column. The right-hand column
shows the term for the corresponding Pauli operator (here a Pauli operator $P(f)$
corresponding to some binary vector $f$ is the product of Pauli $X$ or $Z$ over all qubits in the support of $f$).
The term undetectable error refers to a Pauli operator commuting with all stabilizers.
An undetectable error is called a logical operator if it has a non-trivial action on codewords.
Two logical operators are considered equivalent iff they differ by a product of stabilizers. Equivalent logical operators
have the same action on any codeword.
Note that the sets $\kr{\partial},\kr{\partial^T},\im{\partial},\im{\partial^T}$ are linear spaces, while $\kr{\partial}\backslash \im{\partial}$ and $\kr{\partial^T}\backslash \im{\partial^T}$ are not.  Equivalence classes of logical operators are identified with cosets, that is, elements of the quotient  spaces $\kr{\partial}/\im{\partial}$
and $\kr{\partial^T}/\im{\partial^T}$.

We shall sometimes use the terms `complex' and `code' interchangeably: given a complex, we can define a code in a canonical fashion
as described above. 
The minimum weights of a non-trivial cycle and a non-trivial cocycle coincide with the code distances $d^Z$ and $d^X$ respectively.
Finally, define a {\em homological dimension} of a complex as
\begin{equation}
\label{H}
H(\partial)=\dim{(\kr{\partial})}-\dim{(\im{\partial})}= \dim(\kr{\partial}/\im{\partial}).
\end{equation}
Note that the homological dimension of a complex coincides with the number of logical qubits 
in the corresponding code: $k=n-2\rnk{(\partial)}=(n-\rnk{(\partial)})-\dim{(\im{\partial})}=H(\partial)$.

Let us emphasize that the mapping from complexes to CSS codes is many-to-one. 
Indeed, given a pair of parity check matrices $A^X,A^Z$ as above, one can
define a boundary operator $\partial=(A^Z)^T U A^X$, where $U$ is an arbitrary invertible matrix.
Note that the desired properties $\partial^2=0$, $\im{\partial}=C^Z$, and $\im{\partial^T}=C^X$ 
hold regardless of the choice of $U$. Also note that if we start from a CSS code with low-weight stabilizers,
it is generally not true that $\partial$  is sparse (in the sense of having low-weight rows and columns). 
Finally let us comment that any stabilizer $[[n,k,d]]$ code can be converted  to a CSS
code $[[4n,2k,2d]]$ with $C^X=C^Z$, see Ref.~\onlinecite{BravyiTerhal10maj}.
Moreover, this conversion preserves the stabilizer weight up to a factor $O(1)$. 
In that sense the restriction to CSS codes is not essential.

 For product complexes (defined in the next subsection) we shall reserve
the notation $\partial$ for the boundary operator of the product complex
and denote boundary operators of individual complexes as $\delta_1$ and $\delta_2$.
Unless stated otherwise, below we shall always work with the single sector theory.

\subsection{Product Complex and K\"unneth Formula}

Let $(\calC_1,\delta_1)$ and $(\calC_2,\delta_2)$ be an arbitrary pair of complexes. 
Define an operator
\be
\label{product}
\partial=\delta_1 \otimes I + I \otimes \delta_2
\ee
acting on the tensor product space  $\calC_1 \otimes \calC_2$. Here $I$ is the identity operator. 
We shall always equip the space $\calC_1 \otimes \calC_2$ with the product basis
$i\otimes j$, where $i$ and $j$ are basis vectors of $\calC_1$ and $\calC_2$.
The property $\delta_a^2=0$ implies that $\partial^2=2\delta_1\otimes \delta_2=0$ since we consider 
linear spaces over the binary field (working with more general vector spaces would require
a definition $\partial=\delta_1 \otimes I - I \otimes \delta_2$). 
 Thus $\partial$ is a valid boundary operator. We shall refer to the complex 
 $(\calC_1\otimes \calC_2,\partial)$ as  a product of complexes  $(\calC_1,\delta_1)$ and $(\calC_2,\delta_2)$.

One important property of the product complex is that we can easily compute its homological dimension (the number of logical qubits)
from the ones of individual complexes.
The following simple fact
 is  a special case of the well-known K\"unneth formula, see for instance Ref.~\onlinecite{Hatcher2002}.
\begin{lemma}[\bf K\"unneth formula]
\label{lemma:Kun}
Let $\delta_1,\delta_2$ be any boundary operators and $\partial=\delta_1\otimes I + I\otimes \delta_2$. Then
\begin{equation}
\label{Kun1}
\kr{\partial}=\kr{\delta_1} \otimes \kr{\delta_2} + \im{\partial}
\end{equation}
and
\begin{equation}
\label{Kun2}
H(\partial)=H(\delta_1)\cdot H(\delta_2).
\end{equation}
\end{lemma}
\begin{proof}
Consider any vector $f\in \kr{\partial}$. Define a vector
$g=(\delta_1 \otimes I)f =(I\otimes \delta_2)f$.
By construction,  
$g\in (\im{\delta_1} \otimes \calC_2)\cap (\calC_1 \otimes \im{\delta_2}) = \im{\delta_1}\otimes \im{\delta_2}$,
that is, $g=(\delta_1 \otimes \delta_2) h$ for some $h\in \calC_1\otimes \calC_2$.
Identities $\delta_a^2=0$ then lead to
$(\delta_1 \otimes I)(f+\partial h)=0$ and $(I \otimes \delta_2)(f+\partial h)=0$,
that is,
$f+\partial h \in (\kr{\delta_1}\otimes \calC_2 )\cap (\calC_1 \otimes \kr{\delta_2})= 
\kr{\delta_1}\otimes \kr{\delta_2}$
This proves the inclusion $\subseteq$ in Eq.~(\ref{Kun1}).
The inclusion $\supseteq$ follows trivially from $\delta_a^2=0$ and $\partial^2=0$.
It remains to prove Eq.~(\ref{Kun2}). One can easily check that
 $\im{\delta_1}\otimes \kr{\delta_2}\subseteq \im{\partial}$ and
 $\kr{\delta_1}\otimes \im{\delta_2}\subseteq \im{\partial}$.
 Thus Eq.~(\ref{Kun1}) implies that $\kr{\partial}/\im{\partial}$ has a basis $h_1^i\otimes h_2^j$,
 where $\{h_a^i\}_i$ is a basis of $\kr{\delta_a}/\im{\delta_a}$. This proves Eq.~(\ref{Kun2}).
\end{proof}

Next we compute parameters of the CSS code corresponding to the product complex.
Let $w_a$ be the maximum weight of rows and columns of the boundary operator $\delta_a$. 
Let $d_a^Z,d_a^X$ be the minimum weight of non-trivial cycles or co-cycles, respectively, in the complex $(\calC_a,\delta_a)$.
\begin{lemma}
\label{lemma:Kun1}
Any row and any column of $\partial$ has weight at most $w_1+w_2$.
Furthermore,  let $d^Z,d^X$ be the minimum weight of non-trivial cycles or co-cycles, respectively, in the complex
 $(\calC_1\otimes \calC_2,\partial)$. Then,
\begin{equation}
\label{Kun4}
\max{\{ d_1^\alpha,d_2^\alpha\}}\le d^\alpha \le d_1^\alpha d_2^\alpha, \quad \alpha=X,Z.
\end{equation}
\end{lemma}
We note that the upper bound in Eq.~(\ref{Kun4}) may or may not be tight depending
on the choice of input complexes $(\calC_a,\delta_a)$, see Section~\ref{sec:small} for more details. 
For the product of complexes constructed using the multiple sector theory one
can prove a similar lemma and, moreover, derive simple sufficient conditions
under which the upper bound in Eq.~(\ref{Kun4}) is tight~\cite{InPreparation}.
\begin{proof}[\bf Proof of Lemma~\ref{lemma:Kun1}.]
Consider the case $\alpha=Z$; the case $\alpha=X$ can be handled by replacing $\delta_a$ by $\delta_a^T$.
The matrices $\delta_1\otimes I$ and $I\otimes \delta_2$ have both row and column weights
at most $w_1$ and $w_2$ respectively. By triangle inequality, $\partial$ has row and column weights at most $w_1+w_2$.
 Given any nontrivial cycles $h_a\in \kr{\delta_a}\backslash \im{\delta_a}$,  the vector $h_1 \otimes h_2$ is a nontrivial cycle for $\partial$.
If $h_a$ has weight $d^Z_a$ then $h_1\otimes h_2$ has weight $d^Z_1 d^Z_2$. This proves $d^Z\le d^Z_1 d^Z_2$.
To prove the lower bound on $d^Z$, assume without loss of generality that $d^Z_1\ge d^Z_2$. 
Suppose $\psi$ is a minimum weight
non-trivial cycle for $\partial$. Then we can always choose a 
pair of non-trivial cocycles $h_a\in \kr{\delta_a^T}\backslash \im{\delta_a^T}$
such that $\psi$ and $h_1\otimes h_2$ have odd overlap. 
Using the K\"unneth formula Eq.~(\ref{Kun1}) one can represent
$\psi$ as $\psi=\phi+\theta+\omega$, where 
\[
\phi\in \kr{\delta_1}\otimes \kr{\delta_2}, \quad \theta \in \im{\delta_1}\otimes \calC_2, \quad \omega \in \calC_1\otimes \im{\delta_2}.
\]
Let us identify vectors from $\calC_1\otimes \calC_2$ with matrices
of size $n_1\times n_2$. Using matrix-vector notations we have
$h_1^T \psi h_2 =1$.
Furthermore, $\theta h_2 \in \im{\delta_1}$ and $\omega h_2=0$ since $h_2$ is a cocycle. 
Thus 
$\psi h_2 \in  \phi h_2 + \im{\delta_1}  \subseteq \kr{\delta_1}$.
 On the other hand, $\psi h_2 \notin \im{\delta_1}$
since otherwise $h_1^T \psi h_2=0$. Thus $\psi h_2$ is a non-trivial cycle
for $\delta_1$ and as such it must have weight at least $d^Z_1$.
Since $\psi h_2$ is a linear combination of columns of $\psi$, 
the triangle inequality implies that $\psi$ itself must have weight at least $d^Z_1$.
\end{proof}
As was shown in the previous subsection, 
the complex $(\calC_a,\delta_a)$ describes a CSS code $[[n_a,k_a,d_a,w_a]]$,
where $n_a=\dim{(\calC_a)}$, $k_a=H(\delta_a)$, and $d_a=\min{\{d_a^X,d_a^Z\}}$. Lemmas~\ref{lemma:Kun},\ref{lemma:Kun1} imply that the product 
complex $(\calC_1\otimes \calC_2,\partial)$ describes a CSS code $[[n,k,d,w]]$, where
\be
\label{Kun3}
n=n_1n_2, \quad k=k_1k_2, \quad w=w_1+w_2, \quad d=\min\{d^X,d^Z\},
\ee
and $d^X,d^Z$ are the two distances of the product code which are bounded as in Eq.~(\ref{Kun4}).

{\em Remark:}
The multiple sector version of the product is defined analogously.  In this case, given complexes
with vector spaces ${\cal C}_i$ and ${\cal C}'_i$ and boundary operators $\delta_i$ and $\delta'_i$, define a new complex with spaces
\be
{\cal D}_i=\oplus_j {\cal C}_j \otimes {\cal C}'_{i-j},
\ee
and boundary operators $\partial_i$ defined as follows.  The operator $\partial_i$ has nonzero matrix elements from each space ${\cal C}_j \otimes {\cal C}'_{i-j}$ to the spaces ${\cal C}_{j-1} \otimes {\cal C'}_{i-j}$ and ${\cal C}_j \otimes {\cal C}'_{i-j}$.  The matrix elements to the first space are given by the matrix elements of the operator $\delta_j \otimes I$ while the matrix elements to the second space are given by the matrix elements of the operator $(-1)^j I \otimes \delta'_{i-j}$.
This construction gives an operator $\partial$ such that $\partial_{i-1} \partial_i=0$ for any field.

\section{Random Codes from Random Complexes}
\label{sec:rand}

In this section we define a random ensemble of boundary operators used throughout this paper. 
We will show that the corresponding CSS code is good with high probability.
First we derive a canonical form of a boundary operator.
\begin{lemma}
\label{lemma:canonical}
Consider any complex $(\calC,\delta)$ such that $\delta$ has homological dimension $H$
and rank $L$. Then $\delta=U\delta_0 U^{-1}$, where $U$ is some invertible matrix and 
$\delta_0$ is the canonical boundary operator defined as block matrix
\be
\label{canonical}
\delta_0=\begin{pmatrix} 
0 & 0  & 0\\
0 & 0 & I\\
0 & 0 & 0
\end{pmatrix}.
\ee
Here rows and columns are grouped into blocks of size $H,L,L$.
Furthermore, the number of invertible matrices $U$ such that $\delta=U\delta_0 U^{-1}$
does not depend on $\delta$.
\end{lemma}
\begin{proof}
Let $M=\dim{(\calC)}$.
By definition of the homological dimension, Eq.~(\ref{H}), one has
 $L+H=\dim{(\kr{\delta})}=M-L$, that is, $M=2L+H$.
Choose an arbitrary $H$-dimensional subspace ${\cal H}$ such that 
$\kr{\delta}={\cal H}\oplus \im{\delta}$ is a direct sum. 
Let $I^1,I^2,\ldots,I^{L+H}$ be any basis of $\kr{\delta}$ such that 
$I^1,\ldots,I^H$ span ${\cal H}$ and $I^{H+1},\ldots,I^{H+L}$ span $\im{\delta}$.
Then $I^{H+j}=\delta(I^{H+L+j})$ for some vectors $I^{H+L+1},\ldots,I^{M}$.
Let ${\cal M}$ be the subspace spanned by $I^{H+L+1},\ldots,I^{M}$.
Since $\delta\cdot {\cal M}=\im{\delta}$, $\dim{(\calM)}\le L$, and $\dim{(\im{\delta})}=L$,
one must have $\dim{(\calM)}=L$. The property $\delta^2=0$ implies that
$\calM\cap \kr{\delta}=0$, as otherwise $\delta\cdot \calM$ would have dimension
less than $L$. Thus vectors $I^1,\ldots,I^M$ form a basis of the full
space $\calC$. In this basis $\delta$ has the desired form Eq.~(\ref{canonical}).
Hence $\delta=U\delta_0 U^{-1}$ for some invertible $U$.

To prove the last statement, define a normalizer group $G=\{ U\, :\,  U\delta_0 U^{-1} =\delta_0\}$.
Then $U\delta U^{-1}=V\delta V^{-1}$ implies $V^{-1}U\in G$. Thus for any a given $\delta$
there are $|G|$ invertible matrices $U$ such that $\delta=U\delta_0 U^{-1}$.
\end{proof}
Let us fix $M$ and $H$. Below we consider a random boundary operator $\delta$
distributed uniformly on the set of all $M\times M$ matrices satisfying $\delta^2=0$
and $H(\delta)=H$. By Lemma~\ref{lemma:canonical}, such random boundary operator
can be represented  as $\delta=U\delta_0 U^{-1}$, where $U$
is a random invertible matrix drawn from the uniform distribution. 
Define an {\em encoding rate}
\be
\hfra=H/M.
\ee
We shall be interested in the limit $M,H\to \infty$ such that the encoding rate remains constant.
Let us show that in this limit a random boundary operator gives a code with linear distance with high probability.
\begin{lemma}
\label{lemma:bad}
For any $\epsilon>0$ one can choose $c,\hfra>0$ such that the following is true
for all large enough integers $M$ and for all $H\le \hfra M$.
Let $\delta$ be a random $M\times M$ boundary operator with $H(\delta)=H$.
Then the probability that $\kr{\delta}$ contains a vector with weight less than $cM$ is at most 
$O(1)\cdot 2^{-M/2 + M\epsilon}$ and the same bound also holds for $\ker(\delta^T)$.
\begin{proof}
We just consider the case of $\kr(\delta)$; the proof for $\kr(\delta^T)$ is identical since $\delta$ and $\delta^T$ are drawn from the same distribution.

Let $M=2L+H$. We will say that a vector has low weight iff its weight is less than $cM$. 
By Lemma~\ref{lemma:canonical}, we can assume
that $\delta=U\delta_0U^{-1}$, where $U$ is a random invertible matrix.
Note that  $\kr{\delta}=U\cdot \kr{\delta_0}$. 
For a fixed vector $v\in \kr{\delta_0}$ the rotated vector $Uv$ is
distributed uniformly on the set of all $M$-bit vectors. 
Thus the
 probability 
that $Uv$ has low weight  is equal to
\[
\sum_{w < cM} 2^{-M} {M \choose w}\le O(1)\cdot 2^{-M+S(c)M + o(M)},
\]
where $S(c)=-c\log_2(c) - (1-c) \log_2(1-c)$ is the Shannon entropy. 
The total number of vectors in $\kr{\delta_0}$ is $2^{L+H}=2^{(M+H)/2}$. The union bound implies that
$\kr{\delta}$ contains a low weight vector with probability at most
\[
O(1)\cdot 2^{(M+H)/2 -M+S(c)M +o(M)} = O(1)\cdot 2^{-M/2 + H/2  +S(c)M + o(M)}.
\]
It remains to choose small enough $c$ and $\hfra$ such that $S(c)M+\hfra M/2 + o(M)\le \epsilon M$.

\end{proof}
\end{lemma}

\section{Product of Two Random Complexes: Distance Bounds}
\label{sec:distance}
In this section we study the product of two random complexes $(\calC_a,\delta_a)$ 
defined above. Both complexes have the same dimension, $\dim{(\calC_1)}=\dim{(\calC_2)}=M$,
and the same homological dimension $H=H(\delta_1)=H(\delta_2)=\hfra M$. 
 We prove that for sufficiently small $c>0$ and $\hfra>0$, the product code  has distance at least $cM^2$ with high probability.
The distance bound in the previous section was based on a ``first moment" method: we showed that the average number of low weight 
cycles  is small, implying that with high probability there are no low weight cycles. 
There are two reasons why this kind of estimate will not work for the product code.
One obvious reason is that, by construction, the product code always has cycles with weight $O(M)$.
These are trivial cycles (boundaries) obtained as $\partial (i\otimes j) = (\delta_1 i)\otimes j + i \otimes (\delta_2 j)$,
where $i,j$ are any basis vectors.  Thus some steps in the proof must differentiate between trivial and non-trivial cycles.
The second reason is that, if by chance we pick a poor choice of the boundary operators $\delta_1,\delta_2$ such that $\partial$
has a low weight non-trivial cycle, then in fact $\partial$ will have many low weight non-trivial cycles.
To see this, note that if $\partial $ has a non-trivial cycle $\psi$ with weight $o(M^2)$ then the sum of $\psi$
and any low weight trivial cycle as above is a non-trivial cycle with weight $o(M^2)$. 
 As a result, even though most codes will not have any low weight non-trivial cycles, the average number of such cycles
 will not be small.  This problem motivates our introduction of ``uniform low weight" condition below.

Assume that a vector $\psi$ in the product complex $\calC_1\otimes \calC_2$ exists that is a nontrivial cycle for $\partial$ 
 and has weight less than $cM^2$.  We regard $\psi$ as an $M$-by-$M$ matrix, with rows corresponding to the first complex and columns corresponding to the second.
 Choose any constant $r$ such that  $c<r< 1$. Clearly, $\psi$ has at least 
 $(1-r)M$ columns with weight at most $cMr^{-1}$.  Similarly, 
 $\psi$ has 
 at least $(1-r)M$ rows with weight at most $cMr^{-1}$.  Let $M'=(1-r)M$.  Then, if we consider the $M'$-by-$M'$ submatrix of $\psi$
 consisting just of those rows and columns, then every row has weight at most $cMr^{-1}=c'M'$ where
\be
c'= cr^{-1}/(1-r),
\ee
 and similarly every column also has weight at most $c'M'$.  We refer to this submatrix as the {\em reduced matrix}.  We refer to the condition that an $M'$-by-$M'$ matrix has weight at most $c'M'$ in every row and column as the {\em uniform low weight condition}. The above shows that $\psi$
 must have at least one $M'$-by-$M'$ submatrix obeying the uniform low weight condition. Note that for any fixed $r>0$ one can make $c'$
 arbitrarily small by choosing small enough $c$.

In subsection \ref{first} we show that if each input code has distance at least $M-M'+1$, then
 in the product complex there is no nontrivial cycle which gives a {\it vanishing} reduced matrix.
 The probabilistic estimates from the previous section imply that  for sufficiently large $M'$ the desired distance bound on the input codes will hold with high probability.
The number of possible choices of $M'$ rows out of $M$ is ${M}\choose{M'}$.  Thus, the number of possible choices of $M'$ rows and $M'$ columns is
${{M}\choose{M'}}^2$.   Fix any choice of $M'$ rows and $M'$ columns and let $P_{red}(M')$ denote  the probability that there is a cycle (trivial or nontrivial) which gives a
{\em nonvanishing} reduced matrix  obeying the uniform low weight condition.  
Note that this probability is independent of the particular choice of the set of $M'$ rows and columns.

Summing over all choices of $M'$ rows and columns and using a union bound, the probability that there is a cycle which
contains a non-vanishing reduced matrix obeying the uniform low weight condition
 is bounded by ${M \choose M'}^2 P_{red}(M')$. Thus the probability that there is a nontrivial cycle $\psi$
with weight at most $cM^2$ is bounded by 
\be
{{M}\choose{M'}}^2 P_{red}(M')+o(1),
\ee
where the $o(1)$ accounts for the exponentially small probability that
one of the input codes has distance less than $M-M'+1$.  

The proof of the distance bound for the product code will be based on bounding $P_{red}(M')$.  From here on, when we refer to a reduced matrix, we use the fixed choice of submatrix corresponding to the first $M'$ rows and columns.
To bound $P_{red}(M')$, in  subsection \ref{sec:red}
we estimate the number of different $M$'-by-$M'$ matrices of given rank $R$ which correspond to the reduced matrix of a cycle.  Then in subsection \ref{sec:ulw}, we estimate the probability that an $M'$-by-$M'$ random matrix of given rank obeys the uniform low weight condition.  Combining these two with a union bound, we show that with exponentially high 
probability, there are no cycles which contain a reduced matrix of rank $R\geq 1$ obeying the uniform low weight condition.
Our bounds will be sufficiently tight so that 
${{M}\choose{M'}}^2 P_{red}(M')$ will be bounded by an exponentially small quantity.
Thus the probability that the product code has non-trivial cycle with weight less than $cM^2$ is $o(1)$. 
Since exactly the same bounds apply to cocycles, 
this shows that the product code has distance less than $cM^2$ with probability $o(1)$.
Thus there exist a family of codes $[[M^2,(\hfra M)^2, cM^2, O(M)]]$, as promised in Theorem~\ref{thm:main}.

Some comments on notation: we use $O(...)$ and $o(...)$ notation referring to scaling with $M$.  We work at fixed $\hfra$ throughout, so the big-O notation equivalently refers to scaling with $L$ or $H$.

\subsection{No Vanishing Reduced Matrices}

\label{first}
\begin{lemma}
\label{lemma:1}
Suppose each input code has minimum distance at least $M-M'+1$. 
If  $h\in \kr{\partial}$ is a cycle  with vanishing reduced
matrix, then $h$ is trivial, that is, $h\in \im{\partial}$. The same holds for cocycles. 
\end{lemma}
We shall need the following simple fact proved in  Ref.~\onlinecite{BT09}.
\begin{lemma}[\bf Cleaning Lemma]
\label{lemma:clean}
Suppose a stabilizer code has minimum distance $d$. Let $P$ be any logical operator 
and $S$ be any subset of less than $d$ qubits. Then there exists 
a logical operator $P'$ equivalent to $P$ modulo stabilizers, 
such that $P'$ acts trivially on $S$.
\end{lemma}
Now we can easily prove Lemma~\ref{lemma:1}.
\begin{proof}
Consider any non-trivial co-cycle $\bar{h}_a\in \kr{\delta_a^T}\backslash \im{\delta_a^T}$. Note that  $\bar{h}_a$
represents a logical operator of the $a$-th input code.
 Let $S=\{M'+1,M'+2,\ldots,M\}$.
Since the size of $S$ is less than the code distance, Cleaning Lemma
guarantees  that  there exists
a trivial co-cycle $\bar{\omega}_a\in \im{\delta_a^T}$ such that $\bar{h}_a+\bar{\omega}_i$ has support only
on the interval $[1,M']$. 
Thus we can choose  a basis set of non-trivial  co-cycles
\begin{equation}
\label{clean1}
\kr{\delta_a^T}=\spn{(\bar{h}_a^1,\bar{h}_a^2, ,\ldots,\bar{h}_a^H)}+\im{\delta_a^T} 
\end{equation}
such that $\bar{h}_a^i$ have support only
on the interval $[1,M']$. Let us now choose  basis sets of non-trivial cycles 
dual to the ones defined in Eq.~(\ref{clean1}), that is, 
\begin{equation}
\label{clean2}
\kr{\delta_a}=\spn{(h_a^1,h_a^2,\ldots,h_a^H)}+\im{\delta_a}
\end{equation}
such that 
\begin{equation}
\label{clean3}
(\bar{h}_a^i,h_a^j)=\delta_{i,j}.
\end{equation}
Here $(f,g)=\sum_{p=1}^M f_p g_p$ is the binary inner product between
vectors $f,g$. Applying K\"unneth formula Eq.~(\ref{Kun1}) to $\partial$ and $\partial^T$ one gets
\begin{equation}
\label{clean4}
\kr{\partial}=\spn{\{ h_1^i \otimes h_2^j, \quad 1\le i,j\le H\}} + \im{\partial}
\end{equation}
and
\begin{equation}
\label{clean5}
\kr{\partial^T} = \spn{\{ \bar{h}_1^i \otimes \bar{h}_2^j, \quad 1\le i,j\le H\}} + \im{\partial^T}.
\end{equation}
Suppose now that $h\in \kr{\partial}$ is a cycle with vanishing reduced matrix. 
Using Eq.~(\ref{clean4}), one can write $h$ as
\begin{equation}
\label{clean6}
h=\sum_{i,j=1}^H x_{i,j}\, h_1^i \otimes h_2^j + \omega, 
\end{equation}
for some $\omega\in \im{\partial}$ and some coefficients $x_{i,j}\in \{0,1\}$. 
Since $(\omega,\bar{h}_1^i \otimes \bar{h}_2^j)=0$ for all $i,j$, 
the duality Eq.~(\ref{clean3}) implies that $x_{i,j}=(h,\bar{h}_1^i \otimes \bar{h}_2^j)$.
However, since  $\bar{h}_1^i \otimes \bar{h}_2^j$ has support only on the reduced matrix
and $h$ has vanishing reduced matrix, $x_{i,j}=0$ for all $i,j$. 
This shows that any cycle with vanishing reduced matrix must be trivial. 
\end{proof}

\subsection{Counting Reduced Cycles}

\label{sec:red}

In this subsection we consider a fixed reduced matrix formed by the first $M'$ rows and columns. We say that an $M'\times M'$
matrix $h$ is a {\em reduced cycle} if there exists a full cycle $g\in \kr{\partial}$ such that $g$ contains $h$
in the first $M'$ rows and columns. Let $\Gamma(R)$ be the number of reduced cycles $h$ such that $h$ has rank $R$.
The main goal of this subsection is to derive an upper bound on $\Gamma(R)$. To this end we define a
reduced boundary operator $\partial'$ acting on a properly defined coarse-grained space.
We show that the task of counting  reduced cycles with a given rank is closely related to counting 
matrices in $\kr{\partial'}$ with a given rank.
\begin{definition}
\label{def:good}
A boundary operator $\delta$ is called good iff no non-zero vector in $\kr{\delta}$ 
has support on the last $M-M'$ coordinates. 
\end{definition}
The main result of the subsection is the following.
\begin{theorem}
\label{thm:count}
Suppose the boundary operators $\delta_1,\delta_2$ are good.
Suppose also that $\delta_a$ have homological dimension $H$. 
Let $\Gamma(R)$ be the number of reduced cycles with rank $R$.
Then
\begin{equation}
\label{count:main1}
\Gamma(R)\le O(1)\cdot 2^{(M+H)R - R^2} \quad \mbox{if} \quad R\le H,
\end{equation}
and
\begin{equation}
\label{count:main2}
\Gamma(R)\le O(1)\cdot 2^{(M+H/2)R - R^2/2} \quad \mbox{if} \quad R\ge H.
\end{equation}
Furthermore, $\Gamma(R)$ does not depend on $\delta_a$ as long as $\delta_a$ are good. 
\end{theorem}
In the rest of this subsection we prove the theorem.
 Let $\calC=\spn{\{ 1,2,\ldots,M\}}$ be the full $M$-dimensional binary space.
We begin by defining several subspaces of $\calC$ and linear operators acting on those subspaces.
First, decompose 
\begin{equation}
\label{VV}
\calC=\calV\oplus \calV^>, \quad \calV=\spn{\{ j\, : \, 1\le j\le M' \}}, \quad \calV^>=\spn{\{ j\, : \, M'<j\le M\}}.
\end{equation}
Let $W$ and $W^>$ be projectors onto the sectors $\calV$ and $\calV^>$ in Eq.~(\ref{VV}).
Here by a projector we mean a linear operator on $\calC$ that sends all vectors in one sector
to zero and acts as the identity on the other sector. Thus $W+W^>=I$ is the identity operator on $\calC$.

Let $\delta\, :\, \calC\to \calC$ be the boundary operator describing one of the two input codes.
Recall that $\delta^2=0$.
Define subspaces 
\[
\calS^{>}=W \delta(\calV^>) \subseteq \calV \quad \mbox{and} \quad
\calV'=\calV/\calS^{>}.
\]
By definition, vectors of the quotient space $\calV'$  are
cosets $x+\calS^>$, where $x\in \calV$.  The following lemma defines a {\em reduced boundary operator} $\delta'$
which will play the key role in what follows.
\begin{lemma}
There exists a unique linear operator
$\delta'\, : \, \calV'\to \calV'$ such that  $(\delta')^2=0$ and 
\begin{equation}
\label{reduced}
\delta'(x+\calS^>)=W\delta(x) +S^> \quad \mbox{for any $x\in \calV$}.
\end{equation}
\end{lemma}
\begin{proof}
Let us first show that 
\begin{equation}
\label{tricky}
\im{(W\delta W\delta )}  \subseteq \calS^>.
\end{equation}
Indeed, suppose $x=W\delta W\delta(y)$ for some $y$. 
Then $x=W\delta(I+W)\delta(y)=W\delta W^>\delta(y)\in W\delta(\calV^>)=\calS^>$
which proves Eq.~(\ref{tricky}). To show that Eq.~(\ref{reduced}) indeed defines a linear
operator on $\calV'$ we need to check that the right-hand side of Eq.~(\ref{reduced})
depends only on the coset of $x$. Equivalently, we need to check that $W\delta(\calS^>)\subseteq \calS^>$.
However, this follows from Eq.~(\ref{tricky}) since 
$W\delta(\calS^>)=(W\delta W\delta)(\calV^>)\subseteq \im{W\delta W\delta}$.
Thus $\delta'$ is well-defined. The property
$(\delta')^2=0$ follows trivially from Eq.~(\ref{tricky}).
\end{proof}
We first establish some basic properties of $\delta'$.   Given a vector $h\in \calV$, let
$h'\in \calV'$ be the coset of $h$, that is, $h'=h+\calS^>$.
\begin{lemma}
\label{lemma:basic}
For any vector $g\in \calC$ one has
$(W\delta g)'=\delta' (Wg)'$. Furthermore, 
\begin{equation}
\label{krim}
\kr{\delta'} = \{ (Wg)' \, : \, \delta g\in \calV^> \} \quad \mbox{and} \quad
\im{\delta'}=\{((Wg)'\, : \, g\in \im{\delta} \} .
\end{equation}
\end{lemma}
\begin{proof}
Indeed, $(W\delta g)'=W\delta g + \calS^>= W\delta (W+W^>)g + \calS^>=
W\delta Wg + \calS^>=\delta'(Wg)'$.
Here we used the fact that $W\delta W^>g \in W\delta (\calV^>) = \calS^>$.

Let us show that $\kr{\delta'} = \{ (Wg)' \, : \, \delta g\in \calV^> \}$.
Indeed, suppose $\delta' h=0$. Then the coset $h$ has a representative $f\in \calV$
such that $W\delta f \in \calS^>$, that is, $W\delta (f+k)=0$ for some $k\in \calV^>$.
Let $g=f+k$. Then $\delta g \in \calV^>$ and $h=f+\calS^>=Wg+\calS^>$ proving
that $h=(Wg)'$ has the desired form. Conversely, if $\delta g\in \calV^>$ then
$\delta'(Wg)'=(W\delta g)'=0$ since $W\calV^>=0$. 
The second equality in Eq.~(\ref{krim}) follows trivially from the identity
$(W\delta g)'=\delta' (Wg)'$.
\end{proof}
Recall that we define a homological dimension of a boundary operator $\delta$
as $H(\delta)=\dim{(\kr{\delta})}-\dim{(\im{\delta})}$. Below we show that 
 the boundary operators $\delta$ and $\delta'$ have the same homological dimension,
 as long as $\delta$ is good, see Definition~\ref{def:good}.
 Note that the condition of being good can be rephrased as 
\begin{equation}
\label{goodness}
\ker{\delta} \cap \calV^> =0.
\end{equation}
\begin{lemma}
\label{lemma:good}
Suppose a boundary operator $\delta$ is good. Then $\dim{\calV'} = 2M'-M$ and 
\begin{equation}
\label{good}
\dim{(\kr{\delta'})}=  \dim{(\kr{\delta})}-(M-M'),  \quad \dim{(\im{\delta'})}=  \dim{(\im{\delta})}-(M-M').
\end{equation}
\end{lemma}
\begin{proof}
Since $\dim{\calV'}=M'-\dim{\calS^>}$, it suffices to show that 
$\dim{\calS^>}=M-M'$. By definition, $\calS^>=W\delta(\calV^>)$ 
and thus $\dim{\calS^>}\le \dim{\calV^>}=M-M'$. Suppose
$\dim{\calS^>}< \dim{\calV^>}$. Then there must exist a non-zero vector
$g\in \calV^>$ such that $W\delta(g)=0$. From Eq.~(\ref{goodness})
we infer that $h=\delta(g)\ne 0$ but $Wh=0$, that is, $h\in \calV^>$. This contradicts
to Eq.~(\ref{goodness}) since $\delta h=\delta^2(g)=0$.

 The goodness condition implies that $\delta g\in \calV^>$ is only possible
for $\delta g=0$. Thus the first equality in  Eq.~(\ref{krim}) becomes $\ker{\delta'}=\{ (Wg)' \, : \,  g\in \ker{\delta} \}$.
Noting that $\calS^>\subseteq W(\im{\delta})\subseteq W(\kr{\delta})$ and using Eq.~(\ref{krim}) we arrive at
\[
\dim{(\kr{\delta'})}=\dim{(W\kr{\delta})} - \dim{(\calS^>)} \quad \mbox{and} \quad \dim{(\im{\delta'})}=\dim{(W\im{\delta})} - \dim{(\calS^>)}.
\] 
Using the goodness condition again one can easily show that
$\dim{(W\kr{\delta})}=\dim{(\kr{\delta})}$ and  $\dim{(W\im{\delta})}=\dim{(\im{\delta})}$.
It remains to substitute $\dim{(\calS^>)}=M-M'$.
\end{proof}
The above lemma implies that $H(\delta')=H(\delta)$ whenever $\delta$ is good.
From now on we consider a pair of good boundary operators $\delta_1,\delta_2\, :\, \calC\to \calC$ such that 
\[
\dim{(\im{\delta_a})}=L \quad \mbox{and}\quad \dim{(\kr{\delta_a})}=L+H, \quad \mbox{where} \quad M=2L+H.
\]
Define subspaces $\calS_a^>$ and $\calV_a'$ as above for each boundary operator $\delta_a$. 
Let $\delta_a' \, : \, \calV_a'\to \calV_a'$ be the corresponding reduced boundary operator. 
By Lemma~\ref{lemma:good} we have
\begin{equation}
\label{dimred}
\dim{\calV_a'}=2M'-M\equiv K.
\end{equation}
Consider a tensor product space $\calC\otimes \calC$ and define 
\begin{equation}
\label{partial'}
\partial'=\delta_1'\otimes I + I \otimes \delta_2'
\end{equation}
acting on the space  $\calV_1'\otimes \calV_2'$. Note that 
\begin{equation}
\label{quotient12}
\calV_1'\otimes \calV_2' \cong (\calV\otimes \calV)/\calS_{12}^>,
\quad \mbox{where} \quad  
\calS_{12}^>=\calS_1^>\otimes \calV + \calV\otimes \calS_2^>.
\end{equation}
Given any vector $h\in \calV\otimes \calV$,  let $h'\in \calV_1'\otimes \calV_2'$
be the coset $h+\calS_{12}^>$.  One can easily check that $(f\otimes g)'=f'\otimes g'$ for any $f,g\in \calV$.
The lemma below shows that 
a coset is a cycle for the reduced boundary operator $\partial'$ iff
it has a representative which is a reduced matrix of a cycle for $\partial$.
\begin{lemma}
\label{lemma:basic2}
Suppose $\delta_a$ are good. Then 
\begin{equation}
\label{basic2}
\kr{\partial'}=\{ ((W\otimes W)g)' \, : \, g\in \kr{\partial} \} 
\quad \mbox{and} \quad
\im{\partial'}=\{ ((W\otimes W)g )'\, :\, g\in \im{\partial}\}.
\end{equation}
\end{lemma}
\begin{proof}
Let us first show that $\partial' ((W\otimes W) h)' = ((W\otimes W)\partial h)'$ for any
$h\in \calC\otimes \calC$. 
By linearity, it suffices to consider product vectors $h=g_1\otimes g_2$. 
Then 
\begin{eqnarray}
((W\otimes W) \partial h)' &=& (W\delta_1 g_1)'\otimes (Wg_2)' + (Wg_1)'\otimes (W\delta_2 g_2)'  \nonumber \\
&=& \delta_1' (Wg_1)' \otimes (Wg_2)' + 
(Wg_1)'\otimes \delta_2'(Wg_2)'  \nonumber \\
&=& \partial' ((Wg_1)'\otimes (Wg_2)')= \partial' ((W\otimes W)h)'. \label{WW}
\end{eqnarray}
Here the second equality uses Lemma~\ref{lemma:basic}.
This immediately proves the second equality in Eq.~(\ref{basic2})
and the  inclusion $\kr{\partial'}\supseteq \{ ((W\otimes W)g)' \, : \, g\in \kr{\partial} \}$.

It remains to prove $\kr{\partial'}\subseteq \{ ((W\otimes W)g)' \, : \, g\in \kr{\partial} \}$.
Suppose $\partial' f=0$ for some coset $f\in \calV_1'\otimes \calV_2'$. 
We need to show that $f$ has a representative $g$ which is a reduced matrix of a cycle.
By K\"unneth formula, $\kr{\partial'}=\im{\partial'} +\kr{\delta_1'}\otimes \kr{\delta_2'}$.
By linearity, it suffices to consider two cases. {\em Case~1:} $f\in \im{\partial'}$.
Then the second equality in  Eq.~(\ref{basic2}) implies that $f$ has a representative which is a reduced matrix
of a boundary (and thus a cycle). 
 {\em Case~2:} $f\in \kr{\delta_1'}\otimes \kr{\delta_2'}$.
 Since $\delta_a$ are good, Lemma~\ref{lemma:basic} implies that 
 $\kr{\delta_a'}=\{ (Wg)'\, : \, g\in \kr{\delta_a}\}$. Hence $f$ has a
 representative $g=(W\otimes W) g_{full}$, where $g_{full}\in \kr{\delta_1}\otimes \kr{\delta_2}$.
 Clearly, $f_{full}$ is a cycle and we are done. 
\end{proof}

The first equality in Eq.~(\ref{basic2}) implies that the set of rank-$R$ matrices
of size $M'\times M'$ which are reduced matrices
of cycles  coincides with the set of rank-$R$ matrices $g\in \calV\otimes \calV$
such that the coset $g'$ is a cycle for the reduced boundary operator. Thus 
\begin{equation}
\label{count1}
\Gamma(R)=\sum_{h\in \ker{\partial'}} \#\{ g\in \calV\otimes \calV \, : \, \rnk{(g)}=R \quad \mbox{and} \quad g'=h\}.
\end{equation}
Choose any basis set of cosets $h_a^1,\ldots,h_a^K\in \calV_a'$ and let $g_a^i\in \calV$
be any fixed vector in the coset $h_a^i$. One can always choose a basis of $\calV$ such that 
the first $K$ basis vectors are $g_a^1,\ldots,g_a^K$  and the last $M'-K=M-M'$ basis vectors belong to $\calS_a^>$.  Then 
any vector $g\in \calV\otimes \calV$ in the coset $h$ can be regarded
as an $M'\times M'$ matrix that contains a given $K\times K$ matrix $h$ in the first $K$ rows and
columns.  
\begin{definition}
Let $X$ and $Y$ be arbitrary matrices of size $a\times a$ and $A\times A$ respectively. 
We will say that $Y$ is an extension of $X$ iff $Y$ contains $X$ in the first $a$ rows and columns. 
Let $E_{a,r}^{A,R}$ be the number of rank-$R$ extensions $Y$ of a given rank-$r$ matrix $X$.
\end{definition}
Note that the number of rank-$R$ matrices $Y$ extending a given matrix $X$ is invariant under a
transformation $X\to UXV$, where $U$, $V$ are arbitrary invertible matrices. This means that the number
of rank-$R$ extensions $Y$ depends only on the rank of $X$ and thus the coefficient $E_{a,r}^{A,R}$ is well-defined.
In Appendix~A we prove that 
\begin{equation}
\label{E}
E_{a,r}^{A,R}\le O(1)\cdot 2^{(2A-a)R - ar -R^2 + (r+R)^2/4}
\end{equation}
and
\begin{equation}
\label{Esimple}
E^{A,R}\equiv E_{0,0}^{A,R}=O(1) \cdot 2^{2AR-R^2}.
\end{equation}
Note that $E^{A,R}$ is the total number of rank-$R$ matrices of size $A\times A$.
Using these notations, Eq.~(\ref{count1}) can be written as
\begin{equation}
\label{count1'}
\Gamma(R)=\sum_{r=0}^{\min{\{K,R\}}}  \#\{ h\in \ker{\partial'} \, : \, \rnk{(h)}=r \} \cdot E_{K,r}^{M',R}.
\end{equation}

The remaining step is to compute 
the number of matrices $h'\in \ker{\partial'}$ with a given rank $r$.
This is done in the next lemma; while the lemma is stated in terms of $\partial$, we will
apply it to the reduced boundary operator $\partial'$, using $\dim{(\im{\delta'_a})}=L-(M-M')$.
\begin{lemma}
\label{lemma:count2}
Let $\delta_1,\delta_2$ be boundary operators with  $\dim{(\im{\delta_a})}=L$ and $\dim{(\kr{\delta_a})}=L+ H$.
Define $\partial=\delta_1\otimes I + I \otimes \delta_2$
and let $Z(r)$ be the number of rank-$r$
matrices in $\kr{\partial}$. Then $Z(r)$ is only a function of $r,L,H$ and 
\begin{equation}
\label{Z(R)}
Z(r)\le O(1)\cdot 2^{2(H+L)r-r^2} \cdot \sum_{f=0}^{\min{(r/2,L)}} 
2^{-2f^2 + 2f(r-H)}.
\end{equation}
\end{lemma}
\begin{proof}
By Lemma~\ref{lemma:canonical}, there exist invertible matrices $U_a$
such that a transformation $\delta_a\to U_a\delta_a U_a^{-1}$
brings $\delta_a$ into  the canonical form
 \begin{equation}
\label{bop}
\delta_a=\left[ \begin{array}{ccc} 0 & 0 & 0 \\
0 & 0 & I \\
0 & 0 & 0\\
\end{array} \right],
\end{equation}
where rows and columns are grouped into blocks of size $H,L,L$.
Let $U=U_1\otimes U_2$.
Noting that $(U_1\delta_1 U_1^{-1}) \otimes I + I\otimes (U_2\delta_2 U_2^{-1})=U\partial U^{-1}$ and 
 $\kr{(U \partial U^{-1})} = U\cdot \kr{\partial}$, it suffices to count
rank-$r$ matrices in $\kr{\partial}$ for the special case when both matrices $\delta_a$ have the canonical form.  
Using K\"unneth formula Eq.~(\ref{Kun1}) one can easily check that $\ker{\partial}$ coincides with the set of
matrices $h$ having the following form:
\begin{equation}
\label{bop1}
h=\left[ \begin{array}{ccc} A & B & 0 \\
C & D & F\\
0 & F & 0 \\
\end{array}
\right].
\end{equation}
As above, we group rows and columns into blocks of size $H,L,L$.
Consider the set of matrices $h$ as above where the block $F$ has some fixed rank $f$.
For a fixed choice of $F$ let $S_{row}$ and $S_{col}$ be the set of first $f$ linearly independent
rows and columns of $F$ respectively. Choose any invertible $L\times L$ matrices  $U$ and $V$  
such that $UFV$ has zero rows outside $S_{row}$ and zero columns outside $S_{col}$. 
A transformation 
\[
h\to \left[ \begin{array}{ccc} 
I & 0 & 0 \\
0 & U & 0 \\
0 & 0 & U \\
\end{array} \right] \cdot h \cdot
\left[ \begin{array}{ccc} 
I & 0 & 0 \\
0 & V & 0 \\
0 & 0 & V \\
\end{array} \right]
\]
does not change rank of $h$ and preserves its block structure. Keeping in mind that 
there are $E^{L,f}$  choices of $F$ with a given rank $f$, see Eq.~(\ref{Esimple}), we can now assume that
$F$ has zero rows outside of $S_{row}$ and zero columns outside $S_{col}$. 
Removing all rows of $S_{row}$ and all columns of $S_{col}$ from $h$
reduced its rank by $2f$ regardless of the choice of the remaining blocks $A,B,C,D$. 
After this removal the non-zero part of $h$ forms a matrix of size $(H+L-f)\times (H+L-f)$ which
can be completely arbitrary as long as its rank is $r-2f$. Combining all these observations we arrive at
\begin{equation}
\label{Z(R)1}
Z(r)=\sum_{f=0}^{\min{(r/2,L)}} E^{L,f} \cdot 2^{2f(H+L)-f^2} \cdot E^{H+L-f,r-2f}.
\end{equation}
Here the factor $2^{2f(H+L)-f^2}$ represents possible choices of $A,B,C,D$ in $f$ rows of $S_{row}$
and in $f$ columns of $S_{col}$. Substituting Eq.~(\ref{Esimple})  and collecting similar terms gives Eq.~(\ref{Z(R)}).
\end{proof}
We conclude that the number of reduced cycles with a given rank $R$ is 
\begin{equation}
\label{Gamma2}
\Gamma(R)=\sum_{r=0}^{\min{\{ K,R\} }} Z(r)   \cdot E_{K,r}^{M',R}.
\end{equation}
This shows that $\Gamma(R)$ does not depend on $\delta_a$ as long as $\delta_a$ are good.
From Eq.~(\ref{E}) we get 
\begin{equation}
\label{nasty1}
E_{K,r}^{M',R}=O(1)\cdot 2^{MR-(2M'-M)r - R^2 + (r+R)^2/4}.
\end{equation}
Applying Lemma~\ref{lemma:count2} to the reduced boundary operators $\delta_a'$
and noting that $\dim{(\im{\delta_a'})}=L-(M-M')$, see Eq.~(\ref{good}), 
we can rewrite Eq.~(\ref{Gamma2}) as
\begin{equation}
\label{nasty2}
\Gamma(R)\le O(1)\cdot 2^{MR-3R^2/4}\sum_{r=0}^R 2^{ (H+R/2)r -3r^2/4} \sum_{f=0}^{\infty} 2^{-2f^2 + 2f(r-H)}.
\end{equation}
Here we extended the range of the sum over $f$ in Eq.~(\ref{Z(R)}) to all integers $f\ge 0$ since 
we just need an upper bound on $\Gamma(R)$. 
Likewise, we extended the range of the sum over $r$ in Eq.~(\ref{Gamma2}) to $0\le r\le R$.
The function $2^{-2f^2 + 2f(r-H)}$ has maximum at $f=f_0=(r-H)/2$ and decays exponentially away from $f_0$.
Note that $f_0$ is in the range of the sum over $f$ iff $r\ge H$.  If this is the case, then 
the sum over $f$ can be approximated, up to a factor $O(1)$, by the single term
$2^{-2f_0^2+2f_0(r-H)}=2^{(r-H)^2/2}$.  
In the remaining case, $r<H$, the sum over $f$ 
can be approximated by a constant $O(1)$.
Let $\Gamma_1(R)$ and $\Gamma_2(R)$  be  contributions
to the righthand side of Eq.~(\ref{nasty2}) that come
 from the terms with $r\le H$ and $r\ge H$ respectively.
We have
\begin{equation}
\label{nasty3}
\Gamma_1(R)=O(1)\cdot 2^{MR-3R^2/4} \sum_{r=0}^{\min{\{H,R\}}}
2^{(H+R/2)r - 3r^2/4}.
\end{equation}
The function $2^{(H+R/2)r - 3r^2/4}$ achieves maximum at
$r=r_0=(2/3)H+R/3$ and decays exponentially away from $r_0$.
Note that $r_0\ge \min{\{H,R\}}$ with the equality iff $H=R$.
Hence the sum over $r$ can be approximated, up to a factor $O(1)$, by the 
last term $r=\min{\{H,R\}}$. Simple algebra shows that 
\begin{equation}
\label{nasty4}
\Gamma_1(R)\le O(1) \cdot 2^{(M+H/2)R - R^2/2} \quad \mbox{if $R\ge H$},
\end{equation}
and
\begin{equation}
\label{nasty5}
\Gamma_1(R)\le O(1) \cdot 2^{(M+H)R - R^2} \quad \mbox{if $R\le H$}.
\end{equation}
Next let us bound $\Gamma_2(R)$. Note that terms with $r\ge H$ can only appear for $R\ge H$.
Replacing the sum over $f$ by $O(1)\cdot 2^{(r-H)^2/2}$ in Eq.~(\ref{nasty2}) and simplifying the resulting
expression one gets
\begin{equation}
\label{nasty6}
\Gamma_2(R)=O(1)\cdot 2^{MR-3R^2/4 + H^2/2 } \sum_{r=H}^R 2^{-r^2/4 + Rr/2}.
\end{equation}
The function $2^{-r^2/4 + Rr/2}$ achieves maximum at $r=R$ and decays
exponentially away from the maximum. Approximating the sum over $r$ by the last term
$r=R$, we get
\begin{equation}
\label{nasty7}
\Gamma_2(R)=O(1)\cdot 2^{MR-R^2/2 + H^2/2 }\le O(1)\cdot 2^{(M+H/2)R - R^2/2},
\end{equation} 
since $R\ge H$. 
This proves Eqs.~(\ref{count:main1},\ref{count:main2}).

\subsection{Parameterization of Reduced Cycles}
\label{sec:par}

For any pair of good boundary operators $\delta_1,\delta_2$, let ${\cal Z}_R(\delta_1,\delta_2)$ be the set of $M'$-by-$M'$ matrices which are reduced cycles
and have rank $R$. Note that ${\cal Z}_R(\delta_1,\delta_2)$ has size $\Gamma(R)$, see Theorem~\ref{thm:count}.
Our ultimate goal is to use the union bound to prove that with high probability
(over the choice of $\delta_a$) no matrix in ${\cal Z}_R(\delta_1,\delta_2)$ obeys the uniform low weight condition.
To this end we shall  parameterize reduced cycles in ${\cal Z}_R(\delta_1,\delta_2)$ by integers $j=1,\ldots,\Gamma(R)$
for each good pair $\delta_1,\delta_2$.
Moreover, this parameterization will have  certain symmetry such that for any fixed $j$ and 
for a random pair of good boundary operators $\delta_1,\delta_2$ the $j$-th reduced cycle
in ${\cal Z}_R(\delta_1,\delta_2)$ is distributed uniformly on the set of all rank-$R$ matrices of size $M'$-by-$M'$.
In the rest of this subsection we define a parameterization with the desired symmetry properties.

Below we consider block-diagonal $M\times M$ matrices
\be
\label{Ua}
U_a=\begin{pmatrix}
U_a' & 0 \\
0 & I
\end{pmatrix},
\ee
where $U_1'$ and $U_2'$ are arbitrary invertible $M'\times M'$ matrices. 
Given a pair of good boundary operators $\delta_1,\delta_2$, define
\be
\label{tdelta}
\tilde{\delta}_a=U_a \delta_a U_a^{-1}.
\ee
\begin{lemma}
Suppose $\delta_a$ are good boundary operators. Then $\tilde{\delta}_a$ are also good
boundary operators.
\end{lemma}
\begin{proof}
Let $\delta\equiv \delta_a$, $U\equiv U_a$, and $\tilde{\delta}=U\delta U^{-1}$.
It is clear that $\tilde{\delta}^2=0$, so it suffices to check that $\tilde{\delta}$ is good. 
By definition of goodness, see Eq.~(\ref{goodness}), no 
vector in $\kr{\delta}$ has support on the last $M-M'$ coordinates, that is,
$\kr{\delta}\cap \calV^>=0$. Furthermore, since $U$ is invertible, we have $\kr{\tilde{\delta}} =U\cdot \kr{\delta}$.
Taking into account that $\calV^>=U\cdot \calV^>$ we get
\[
\kr{\tilde{\delta}}\cap \calV^> = (U\cdot \kr{\delta}) \cap (U\cdot \calV^>) = U \cdot (\kr{\delta} \cap \calV^>) =0.
\]
Hence $\tilde{\delta}$ is good.
\end{proof}
The following lemma provides the desired parameterization of reduced cycles.
\begin{lemma}
\label{lem:param}
One can parameterize reduced cycles in each set ${\cal Z}_R(\delta_1,\delta_2)$
by integers $j=1,\ldots,\Gamma(R)$ such that the following properties hold.
(1) The parameterization is defined for any good pair $\delta_1,\delta_2$. (2)
Choose random boundary operators $\delta_1,\delta_2$ from the uniform distribution.
Conditioned on  $\delta_1,\delta_2$ being good, the $j$-th reduced cycle in ${\cal Z}_R(\delta_1,\delta_2)$
is distributed uniformly on the set of all $M'\times M'$ matrices with rank $R$.
\end{lemma}
\begin{proof}
Define $\tilde{\partial}=\tilde{\delta}_1\otimes I + I \otimes \tilde{\delta}_2$.
Noting that $\tilde{\partial}=(U_1\otimes U_2)\partial (U_1\otimes U_2)^{-1}$ one easily gets
\begin{equation}
\label{tilde}
\kr{\tilde{\partial}}=(U_1\otimes U_2) \kr{\partial}.
\end{equation} 
Suppose $g\in \calV\otimes \calV$ is a reduced cycle for $\partial$, that is,
$g=(W\otimes W)g_{full}$ for some full cycle $g_{full}\in \kr{\partial}$.
Let $\tilde{g}=(U_1'\otimes U_2') g$. 
Taking into account that $U_a'W=WU_a$ we get 
$\tilde{g}=(W\otimes W) (U_1\otimes U_2)g_{full}$. Since $(U_1\otimes U_2)g_{full}\in \kr{\tilde{\partial}}$,
see Eq.~(\ref{tilde}), we conclude that 
$\tilde{g}$ is a reduced cycle for $\tilde{\partial}$, that is,
$\tilde{g}\in {\cal Z}_R(\tilde{\delta}_1,\tilde{\delta}_2)$. 
The same argument shows that $(U_1'\otimes U_2')^{-1}\tilde{g}$ is a reduced cycle
for $\partial$ whenever $\tilde{g}$ is a reduced cycle for $\tilde{\partial}$. Hence 
\begin{equation}
\label{tilde1}
 {\cal Z}_R(\tilde{\delta}_1,\tilde{\delta}_2)=(U_1'\otimes U_2') {\cal Z}_R(\delta_1,\delta_2).
\end{equation}
Consider some fixed good pair $\delta_1,\delta_2$.
By theorem~\ref{thm:count}, the set ${\cal Z}_R(\delta_1,\delta_2)$ has size $\Gamma(R)$. 
Choose an arbitrary parameterization of the set ${\cal Z}_R(\delta_1,\delta_2)$ by integers $j=1,\ldots,\Gamma(R)$.
Let $\psi_j$ be the $j$-th reduced cycle in ${\cal Z}_R(\delta_1,\delta_2)$. 
Then consider all possible pairs $\tilde{\delta}_1,\tilde{\delta}_2$ as defined in Eqs.~(\ref{Ua},\ref{tdelta}) and 
choose $(U_1'\otimes U_2') \psi_j$ as the $j$-th
reduced cycle of  ${\cal Z}_R(\tilde{\delta}_1,\tilde{\delta}_2)$.
By Eq.~(\ref{tilde1}), this parameterizes  the sets ${\cal Z}_R(\tilde{\delta}_1,\tilde{\delta}_2)$.
Next choose any good pair  $\delta_1,\delta_2$ which has not been considered yet.
Choose an arbitrary parameterization on the set  ${\cal Z}_R(\delta_1,\delta_2)$
and extend it to all sets  ${\cal Z}_R(\tilde{\delta}_1,\tilde{\delta}_2)$ as described above. 
Repeating these steps we can parameterize the sets ${\cal Z}_R(\delta_1,\delta_2)$ for all good pairs $\delta_1,\delta_2$.

It remains to note that we choose boundary operators from a distribution invariant under
the transformation $\delta_a\to \tilde{\delta}_a$. Hence the distribution of the $j$-th reduced
cycle $\psi_j$ is invariant under a transformation $\psi_j\to (U_1'\otimes U_2')\psi_j$,
where $U_a'$ are arbitrary invertible matrices. This is only possible if $\psi_j$
is distributed uniformly on the set of all rank-$R$ matrices.
\end{proof}

\subsection{Probability of Having Uniform Low Weight}
\label{sec:ulw}

In this subsection we derive an upper bound on the probability that a random
rank-$R$ matrix  has low weight in all rows and in all columns
(uniform low weight condition). To simplify notation,  the lemma below
is stated for $M\times M$ matrices. However, it should be kept in mind that the lemma
will be applied 
to reduced matrices of cycles which have size $M'\times M'$.

\begin{lemma}
\label{probbound}
For any $\epsilon>0$ one can choose $c>0$ such that the following is true
for all integers $1\le R\le M$. 
Let $Z$ be a random rank-$R$ matrix of size $M\times M$  drawn from the uniform distribution
on the set of such matrices.
Then the  probability that every row and every column of $Z$ has weight at most $cM$ is upper bounded by
\be
\label{probeqbnd}
O(1)\cdot 2^{R^2-2(1-\epsilon)M R}.
\ee
\end{lemma}
\begin{proof}
Let $A,B$ be random rank-$R$ matrices of size $M\times R$ drawn from the uniform distribution on the set of such matrices.
Then $Z=AB^T$ is uniformly distributed on the set of rank-$R$ matrices of size $M\times M$.
Below we fix some pair $A,B$ and define two submatrices of $Z$; one of size $M\times R$ and the other of size
$R\times M$. 

Since $Z$ has rank $R$, one can choose an $M\times R$ submatrix of $Z$ which has rank $R$.
Let $Z_{red}$ be any such submatrix. Since each column of $Z$ is a linear combination of columns of $A$,
we conclude that $Z_{red}=AU$ for some invertible $R\times R$ matrix $U$. 
For each matrix $A$ as above let $A_{red}$ be some fixed $R\times R$ submatrix of $A$ with rank $R$
(say, order all $R\times R$ submatrices of $A$ lexicographically and choose $A_{red}$ as the first submatrix with rank $R$). 
Note that $A_{red} B^T$ is a submatrix of $Z$ which has size $R\times M$. 

Let say that a vector has {\em low weight}  if the fraction of non-zero entries in this vector is at most $c$.
Define three classes of matrices. A matrix is Column-Low-Weight (CLW) if each of its columns has low weight. A matrix
is Row-Low-Weight (RLW) if each of its rows has low weight. Finally, a matrix is Column-Row-Low-Weight (CRLW)
if it is both CLW and RLW.
Our goal is to bound the probability that $Z=AB^T$  is CRLW. 
Clearly, if $Z$ is CRLW then any $R\times M$ submatrix of $Z$ must be RLW and
any $M\times R$ submatrix of $Z$ must be CLW. 
The above arguments and the union bound imply that
\begin{equation}
\label{ld1}
\mathrm{Pr}{\mbox{[ $AB^T$  is CRLW ] }} \le \sum_U \mathrm{Pr}{\mbox{[ $A_{red}B^T$ is RLW and $AU$ is CLW ] }},
\end{equation}
where the sum runs over all $R\times R$ invertible matrices $U$. Note that the number of such matrices
is at most $2^{R^2}$. Furthermore, for any fixed $A$ the matrix $A_{red}B^T$ is distributed uniformly
on the set of all $R\times M$ matrices of rank $R$. Likewise, for any fixed $U$ the matrix 
$AU$ is distributed uniformly on the set of all $M\times R$ matrices of rank $R$. This shows that
\begin{equation}
\label{ld2}
\mathrm{Pr}{ \mbox{[ $AB^T$  is CRLW ] }} \le 2^{R^2} \cdot \mathrm{Pr}{\mbox{[ $B^T$ is RLW ]}} \cdot   \mathrm{Pr}{\mbox{[ $A$ is CLW ]}}.
\end{equation}
Let us show that for any $\epsilon>0$ one can choose $c>0$ such that 
\begin{equation}
\label{ld3}
\mathrm{Pr}{\mbox{[ $A$ is CLW ]}}\le O(1)\cdot 2^{-MR(1-\epsilon)}
\end{equation}
for all integers $1\le R\le M$.
Indeed, let $\tilde{A}$ be a random $M\times R$ matrix drawn form the uniform distribution on the set
of all such matrices. Since columns of $\tilde{A}$ are independent and uniformly distributed, one can easily check that
\begin{equation}
\label{ld4}
\mathrm{Pr}{\mbox{[ $\tilde{A}$ is CLW ]}}\le 2^{-MR(1-\epsilon)},
\end{equation}
where $\epsilon$ can be made arbitrarily small by choosing small enough $c$.
On the other hand, $A$ and $\tilde{A}$ have the same distribution conditioned on the event
$\rnk{(\tilde{A})}=R$. Thus 
\begin{equation}
\label{ld5}
\mathrm{Pr}{\mbox{[ $A$ is CLW ]}}=\mathrm{Pr}{\mbox{[ $\tilde{A}$ is CLW $|$ $\rnk{(\tilde{A})}=R$  ]}}
\le \frac{\mathrm{Pr}{\mbox{[ $\tilde{A}$ is CLW ]}}}{\mathrm{Pr}{\mbox{[ $\rnk{(\tilde{A})}=R$ ]}}}
\end{equation}
It is well-known that a random uniformly distributed matrix has
full rank with probability $\Omega(1)$.  Thus  the denominator in Eq.~(\ref{ld5}) is $\Omega(1)$.
Combining Eqs.~(\ref{ld4},\ref{ld5}) proves Eq.~(\ref{ld3}). 
Applying exactly the same arguments to $B$ one can show that 
$\mathrm{Pr}{\mbox{[ $B^T$ is RLW ]}}\le O(1)\cdot 2^{-MR(1-\epsilon)}$.
The lemma now follows from Eq.~(\ref{ld2}).
\end{proof}

\subsection{Union Bound}
\label{third}

In this subsection we combine all ingredients developed above to complete the proof of Theorem~\ref{thm:main}.
\begin{lemma}
\label{fixedrowscolumns}
For any $\eta>0$ one can choose constants $r,c,\hfra>0$ such that the following is true for all
sufficiently large integers $M\ge 1$ 
and $M'=(1-r)M$. Let $\delta_1,\delta_2$ be random $M\times M$ boundary operators
with the homological dimension $H= M\hfra$. Let $P_{red}(M')$ be the probability that 
there exists a nonzero reduced cycle obeying the uniform low weight condition
with a constant $c$. Then 
\begin{equation}
\label{ubound}
 P_{red}(M')\le  O(1)\cdot 2^{-M(1-\eta)/2}
 \end{equation}
\begin{proof}
Let $P^{bad}(M)$ be the probability that $\delta_1$ or $\delta_2$ is not good. Then 
\begin{equation}
\label{badgood}
P_{red}(M')\le P_{red}^{good}(M')+P^{bad}(M),
\end{equation}
where $P_{red}^{good}(M')$ is the probability that 
there exists a nonzero reduced cycle obeying the uniform low weight condition
with a constant $c$  conditioned on both boundary operators $\delta_1,\delta_2$ being good. 
Lemma~\ref{lemma:bad} guarantees that 
\begin{equation}
\label{Pbad}
P^{bad}(M)\le O(1)\cdot 2^{-M/2 + \eta M/2}
\end{equation}
for small enough constants $r,\hfra$.
Let us now bound $P_{red}^{good}(M')$.
By lemma \ref{lem:param}, for random good $\delta_1,\delta_2$, the $j$-th reduced cycle in ${\cal Z}_R(\delta_1,\delta_2)$ is distributed uniformly on the
set of matrices with rank $R$. 
For any fixed $r$, we can choose a $c$ such that $c'=cr^{-1}/(1-r)$ is arbitrarily small.  
Hence, by Lemma~\ref{probbound}, for any $r<1$ and for any $\epsilon>0$, there is a $c$
such that the probability that the $j$-th  reduced cycle in ${\cal Z}_R(\delta_1,\delta_2)$ obeys the uniform low weight condition
with the constant $c$ is upper bounded by
\be
O(1)\cdot 2^{R^2-2(1-\epsilon)M' R}.
\ee

By Theorem \ref{thm:count}, the number of reduced cycles in  ${\cal Z}_R(\delta_1,\delta_2)$  is bounded by
\be
\label{count:1}
\Gamma(R)\le O(1)\cdot 2^{(M+H)R - R^2} \quad \mbox{if} \quad R\le H,
\end{equation}
and
\begin{equation}
\label{count:2}
\Gamma(R)\le O(1)\cdot 2^{(M+H/2)R - R^2/2} \quad \mbox{if} \quad R\ge H.
\end{equation}
Applying the union bound to account for all $\Gamma(R)$ reduced cycles we get
\be
\label{Pred1}
P_{red}^{good}(M')\le O(1)\cdot 
\sum_{R=1}^{M'} \Gamma(R)\cdot 
2^{R^2-2(1-\epsilon)M' R}.
\ee

We break the sum over $R$ into the sum over $1 \leq R < H$ and the sum over $H \leq R \leq M'$.
The sum over $1\le R\le H$ in Eq.~(\ref{Pred1}) is upper bounded by
\begin{equation}
O(1)\cdot \sum_{R \geq 1} 2^{(M+H)R - R^2} \cdot 2^{R^2-2(1-\epsilon)M' R}  
=O(1)\cdot \sum_{R \geq 1} 2^{-MR(1 - \hfra - 2\epsilon' )},
\end{equation}
where $\epsilon'=1-(1-\epsilon)(1-r)$.
The last sum can be upper bounded, up to a factor $O(1)$, by its first term $O(1)\cdot 2^{-M(1-\hfra-\epsilon')}$.
For sufficiently small $r,\epsilon,\hfra$
this is upper bounded by 
$O(1) \cdot 2^{-M/2}$.
The sum over $H \leq R \leq M'$ in Eq.~(\ref{Pred1}) is upper  bounded by
\begin{equation}
\label{qform2}
O(1)\cdot \sum_{R=H}^{M'} 2^{(M+H/2)R - R^2/2} \cdot 2^{R^2-2(1-\epsilon)M' R}  
\le O(1)\cdot \sum_{R=H}^{M'} 2^{-MR(1/2-\sigma)},
\end{equation}
where $\sigma=2\epsilon'+\hfra/2$. Here we used a trivial bound $2^{R^2/2}\le 2^{RM/2}$.
Note that $\sigma$ can be made arbitrarily small by choosing small enough $r,\hfra,c$.
Let us make $\sigma<1/2$. Then 
the sum over $R$ in Eq.~(\ref{qform2}) can be upper bounded, up to a factor $O(1)$, by its first term
$O(1)\cdot 2^{-MH(1/2-\sigma)}\equiv F(M,H)$. 
Note that for any fixed $M$ the function $F(M,H)$ is monotone decreasing for $H\ge 0$. 
Since $H$ is a non-negative integer, one has $F(M,H)\le F(M,1)=O(1)\cdot 2^{-M(1/2-\sigma)}$.
If we   make $\sigma<\eta/2$ then $F(M,H)\le O(1)\cdot 2^{-M/2 + \eta M/2}$.
Combining the two contributions  to the sum in Eq.~(\ref{Pred1})
 we arrive at
$P_{red}^{good}(M')\le O(1) \cdot 2^{-M/2 + M\eta/2}$.
The lemma now follows from Eqs.~(\ref{Pbad},\ref{badgood}).
\end{proof}
\end{lemma}

\begin{theorem}
For sufficiently small $c$ and $\hfra$, the probability that the distance of the product code is less than $cM^2$ is $o(1)$.
\begin{proof}
For any fixed choice of the reduced matrix the probability of having a cycle with a non-vanishing reduced matrix
obeying the uniform low weight condition is $P_{red}(M')$ which is at most $O(1)\cdot 2^{-M/2 + \eta M/2}$, see Lemma~\ref{fixedrowscolumns}.
The number of possible choices of rows and columns for the reduced matrix is ${M \choose M'}^2$. 
A union bound  implies that the probability of having a cycle with a non-vanishing reduced matrix 
 obeying the uniform low weight condition for some choice of rows and columns is bounded by
${M \choose M'}^2 P_{red}(M')$, and for sufficiently small $r,c,\hfra$ this probability is $o(1)$.
By Lemma~\ref{lemma:1}, if the input codes have distance at least $M-M'+1$, then there is no nontrivial cycle which gives a vanishing reduced matrix.
By  Lemma~\ref{lemma:bad}, the probability that the input codes have distance at least $M-M'+1$ is $1-o(1)$ for sufficiently small $r$.
By a union bound, the probability that there is a nontrivial cycle
with weight less than $cM^2$ is $o(1)$ for sufficiently small $c$ and $\hfra$.

To lower bound the distance of the product code, it is also necessary to lower bound the weight of a nontrivial cocycle.  The proof of this is identical to the proof to the weight bound for a nontrivial cycle since $\partial$ and $\partial^T$ are drawn from the same distribution.  So, the probability of a nontrivial cocycle with weight less than $cM^2$ is also $o(1)$  for sufficiently small $c$ and $\hfra$. The theorem follows  by a union bound.

Note that some of the intermediate lemmas, such as lemma \ref{lemma:1} and lemma \ref{fixedrowscolumns} depend upon the particular choice of $r=(M'-M)/M$; however, once we have chosen a sufficiently small $r$, the result holds for all sufficiently small $c$ and $\hfra$ and so $r$ does not enter into the statement of the theorem.
\end{proof}
\end{theorem}

\section{Numerical Results on Small Codes}
\label{sec:small}

In this section we apply the homological product operation to combine two small codes correcting a single error
and produce larger codes correcting multiple errors. Since the same task can be accomplished by code concatenation,
a natural question is how the two methods compare with each other. We will see that  sometimes the homological product
and  concatenation produce codes with the same parameters $[[n,k,d]]$, but the homological product
leads to stabilizers with smaller weight. 

\subsection{Homological Product of Two $[[7,1,3]]$ codes}
\label{sub:7}
To simplify notations, we shall restrict consideration to CSS codes satisfying $C^Z=C^X$
known as self-orthogonal codes. A generalization to arbitrary CSS codes is straightforward.

Let $C=C^Z=C^X\subseteq \FF_2^n$ be the parity check space describing some self-orthogonal
CSS code $\calQ=[[n,k,d]]$. Note that $C\subseteq C^\perp$ by definition of a CSS code.
 We shall assume that $C$ is equipped with some fixed
basis with basis vectors $a^1,\ldots,a^m\in C$.
The  code $\calQ$ can be described by a complex
with a boundary operator $\delta\, : \, {\mathbb{F}}_2^n \to {\mathbb{F}}_2^n$
satisfying $\im{\delta}=C$ and $\im{\delta^T}=C$,
see Section~\ref{subs:CSS} for details.
 A general solution $\delta$ of these equations has a form
\begin{equation}
\label{Udelta}
\delta=\sum_{i,j=1}^m U_{i,j} \, a^i (a^j)^T,
\end{equation}
where $U$ is an arbitrary invertible $m\times m$ matrix.  This shows that a mapping from
a CSS code to a boundary operator is not unique. The freedom in choosing matrix $U$
in Eq.~(\ref{Udelta}) roughly correspond to a freedom in choosing a basis set of 
parity checks (stabilizers) for a given code.  We will see below that this freedom can be exploited to 
obtain better products codes.

Define a product complex with a boundary operator
$\partial=\delta_1\otimes I + I\otimes \delta_2$ acting on $ {\mathbb{F}}_2^n \otimes  {\mathbb{F}}_2^n$,
where
\begin{equation}
\label{UVdelta}
\delta_1=\sum_{i,j=1}^m U_{i,j} \, a^i (a^j)^T \quad \mbox{and} \quad \delta_2=\sum_{i,j=1}^m V_{i,j} \, a^i (a^j)^T.
\end{equation}
Here $U$ and $V$ are arbitrary invertible matrices. 
Applying Eqs.~(\ref{Kun3},\ref{Kun4}) we conclude that  the product complex describes a CSS code $[[n^2,k^2,d']]$,
where $d\le d'\le d^2$ may depend on the choice of $U,V$. Note that the product code
is self-orthogonal whenever $U$ and $V$ are symmetric matrices, $U^T=U$ and $V^T=V$.
Indeed, in this case $\delta_a^T=\delta_a$ and thus $\partial^T=\partial$.  The property of being self-orthogonal
may be useful in fault-tolerance applications, since it enables transversal application of the Hadamard gate. 

Let us now apply this construction to the Steane code, $\calQ=[[7,1,3]]$. 
We choose basis parity checks $a^1,a^2,a^3\in  {\mathbb{F}}_2^7$  as columns of the following matrix:
\begin{equation}
\label{eq:A}
A=
{ \scriptsize \begin{array}{|c|c|c|}
\hline
1 & 0 & 0 \\
\hline
0 & 1 & 0 \\
\hline
0 & 0 & 1\\
\hline
0 & 1 & 1\\
\hline
1 & 0 & 1 \\
\hline
1 & 1 & 0 \\
\hline
1 & 1 & 1\\
\hline
\end{array}}.
\end{equation}
It is well-known that all non-zero vectors in the codespace $C=\mathrm{span}(a^1,a^2,a^3)$ 
have weight $4$. Since any column and any row of $\delta_a$ belongs to $C$,
it follows that any row and any column of $\partial=\delta_1\otimes I + I\otimes \delta_2$ has weight at most $8$
regardless of the choice of $U,V$.

We computed the distance of the product code described by $\partial$ numerically by performing an exhaustive search 
over all non-trivial cycles and co-cycles. Note that there are $2\times 2^{(49-1)/2}=2^{25}$ cycles and cocyles 
to be examined. 
The distance was computed for all possible choices of $U$ and for $V=I$.
We observed that the product code always has parameters  $[[49,1,7]]$ or $[[49,1,9]]$. 
 The first case occurs  if and only if $U$ is a symmetric matrix, $U^T=U$.
While this might be merely a coincidence, we note that the product code has distance $9$ whenever it is not self-orthogonal.
In all cases stabilizers of the product code have weight at most $8$.
The product code $[[49,1,7]]$ is an example when the upper bound on the distance in Lemma~\ref{lemma:Kun1} is not tight. 
In the general case when both $U$ and $V$ are arbitrary invertible matrices,
one can show~\cite{InPreparation} that the product code has parameters $[[7,1,3]]$ if and only if
$\delta_1=\sigma \delta_2^T \sigma^{-1}$ for some $7\times 7$ permutation matrix $\sigma$.

For comparison, concatenating the Steane $[[7,1,3]]$ code with itself produces $[[49,1,9]]$ code with stabilizers of weight $12$.
Indeed, any non-trivial stabilizer of the Steane code has weight $4$.  Concatenation  replaces each single-qubit Pauli operator 
by a three-qubit logical Pauli operator which produces stabilizers with weight $12$.

\section{Open Problems}
\label{sec:openproblems}
In this section, we discuss certain open problems.  For some of these problems, we give partial results and sketches of a solution, while others are left completely open.

\subsection{Higher Powers and Weight Reduction}
The most obvious open quetion is whether similar distance bounds can be proven for higher powers.  In this case, we must overcome the same obstacles as in the case of the product of two codes, that there are low weight cycles by construction and that the first moment method does not work in its simplest form.  However, we lose many of the advantages of working with matrices and are instead forced to work with higher rank tensors.  In addition, the reduced matrix approach will need modification as the most natural ``reduced tensor" ideas do not work.  Thus, while we conjecture that we maintain linear distance for a product of $O(1)$ random codes, this may be rather difficult to prove.

If, however, we were able in this fashion to construct codes of linear distance with generators of weight $O(n^\alpha)$, then for sufficiently small $\alpha$ we would likely be able to use these codes to produce stabilizer codes on $n$ qubits with stabilizer weight $O(1)$ and with distance $\Omega(n^{1/2+\epsilon})$ for some $\epsilon>0$.  This would be accomplished using an idea of weight reduction that we now sketch.  This weight reduction idea is distinct from the idea in Ref.~\onlinecite{bacon}; in particular, it would still produce a stabilizer code rather than a subsystem code.  Before discussing the weight reduction idea, we give some topological motivation for the idea.  Some may prefer to skip the remainder of this paragraph as well as the next paragraph, and proceed to the paragraph after that which gives a way of reducing weight that acts directly on the code without introducing a cell complex or manifold.  The ideas following in this subsection were obtained in discussion with M. Freedman and we only present a very brief sketch here.
Given a CSS code, we can construct a multiple sector chain complex.  For reasons that will become clear, instead of labeling sectors $2,1,0$ we label them by $3,2,1$ so that the vector spaces are ${\cal C}_3,{\cal C}_2, {\cal C}_1$, with boundary operators $\partial_3,\partial_2$.  Unlike the single sector case, this multiple sector chain complex can be constructed in a canonical fashion from the code, up to an arbitrary permutation of the stabilizers: each column of the boundary operator $\partial_3$ corresponds to a given $Z$ stabilizer while each row of $\partial_2$ corresponds to a given $X$ stabilizer.
We ask the question $*$: given this chain complex, can we construct a manifold and a cellulation of that manifold that gives rise to that chain complex, or to a chain complex with the same homology and with the same distance and weight up to $\Theta(1)$ factors?
At this point, the reason for using sectors $3,2,1$ rather than $2,1,0$ becomes clear: if we had used sectors $2,1,0$, then since a $1$-cell can have at most two $0$-cells in its boundary, we would be restricted to the case that each qubit had at most two $X$ stabilizers supported on that qubit and so the answer to the question $*$ would have been negative for many codes for purely local reasons.  However, using sectors $3,2,1$ we can give $*$ a positive answer.  Define a cell complex with a single $0$-cell.  One then attaches $1$-cells to this $0$-cell in closed loops, with one $1$-cell per $X$ stabilizer.  Then, $2$-cells are attached to the $1$-cells, using $\partial_2$ to define the attachment.  Similarly, one attaches $3$-cells to the $2$-cells, using $\partial_3$ to define the attachment.  This gives a $3$-complex with boundary; one can then construct a $7$-manifold without boundary by embedding this complex in general position in high enough dimension to avoid intersection, thickening it, and then attaching another copy of the manifold and identifying the boundaries.

While this question $*$ might be interesting for topological reasons (for example, it enables us to turn interesting codes into interesting manifolds), it also suggests a useful weight reduction procedure for the code.  The cell complex we produce may be quite complicated locally, as each cell attaches to many other cells if the stabilizers are high weight.  However, we can refine the cellulation of the manifold until we have a cellulation with bounded local geometry, so that each cell attaches only to a bounded number of other cells.  This refined cellulation then defines a new code; it has the same number of encoded qubits as the original code because the homology has not changed, and by construction the weight of the generators is now $O(1)$.  However, we have increased the number of qubits by increasing the number of cells, and we have possibly changed the distance of the code, and so a detailed analysis is needed to see how these change.

Rather than giving this detailed analysis in the language of manifolds, we present a procedure to reduce weight that acts directly on codes.  The basic step in the reduction procedure is a ``splitting step".  This procedure is inspired by the idea of refining a cellulation; one can see that, for example, the splitting step we introduce for $Z$ stabilizers corresponds to the step of taking a $3$-cell and refining it into two $3$-cells by adding an additional $2$-cell.

There are two types of splitting steps, $Z$-type and $X$-type.
The $Z$-type splitting step acts as follows on a CSS code with $n$ qubits and $n_Z$ $Z$-type stabilizers and $n_X$ $X$-type stabilizers.  Let $S$ denote some chosen $Z$ stabilizer.  This $Z$ stabilizer can be written as the product of Pauli operators $S=\prod_{i \in T} S^z_i$, where $T$ is some subset of qubits with $|T|$ equal to the weight of the stabilizer.  Pick a set $T_1 \subset T$, with $|T_1| \approx |T|/2$.  Then, define a new code as follows.  The new code has $n+1$ qubits.  We label the qubits as $1,...,n$ and by $a$, where $a$ is a new qubit added to the code.  The new code has $n_Z+1$ $Z$-type stabilizers.  Of these, $n_Z-1$ of them are given by the $Z$-type stabilizers of the original code {\it other} than stabilizer $S$.  The other two stabilizers are
\be
S^z_a \prod_{i \in T_1} S^z_i \quad \; \quad 
S^z_a \prod_{i \in T_2} S^z_i.
\ee
There are $n_X$ $X$-type stabilizers.  For each $X$-type stabilizer $R$ of the old code, we define an $X$-type stabilizer $R'$ of the new code, with $R'=R$ if $R$ commutes with 
$ \prod_{i \in T_1} S^z_i$ and with $R'=S^x_a R'$ if $R$ anti-commutes with  $\prod_{i \in T_1} S^z_i$.  This completes the description of the $Z$-type splitting step.
The $X$-type splitting step is defined analogously, with $X$ and $Z$ interchanged.

We see that the splitting step roughly halves the weight of one stabilizer, while increasing the weight of some other stabilizers by $1$.  A useful reduction procedure would be to first reduce the weight of all $Z$-type stabilizers to $O(1)$; then reduce the weight of all $X$-type stabilizers to $O(1)$; then repeat until all weights are $O(1)$.  Alternately, one could simply applying the splitting step to randomly chosen stabilizers of either $Z$-type or $X$-type.
A brief heuristic analysis (not given here) suggests that for a code given by a homological product of random codes, if we start with a code with $n$ qubits and weight $O(n^\alpha)$, then this procedure will terminate, and for all $\epsilon>0$, there is an $\alpha$ such that it produces a code on $n'$ qubits with stabilizers of weight $O(1)$ and distance $\Theta((n')^{1-\epsilon})$.  We leave this also as an open problem.

\subsection{Applications to Quantum Memories}
Another question is applications of these codes.  While the asymptotic scaling is of interest, early hardware implementations of quantum memories using stabilizer codes will likely be restricted to small number of qubits.  Here, small codes such as the $[[49,1,9,8]]$ CSS code that we found might be of interest.  Note that any $[[n,1,9]]$ CSS code must obey $n \geq 35$, as shown in Ref.~\onlinecite{lowerbounds} and thus the $[[49,1,9]]$ code is close to optimal for a CSS code with the given distance, while keeping the stabilizer weight small. 
The analysis of the performance of these small codes under random noise is an open problem.

Another possible code to consider is the homological product of the toric code with a code such as the $[[7,1,3]]$ CSS code or other code with $n=O(1)$.  In this case, it will be necessary to consider the multiple sector version of the product as the toric code arises from such a multiple sector chain complex (one can write the toric code in terms of a single sector, but the boundary operator is no longer sparse in this case).  One might hope that such a product would improve the error correction properties of the toric code while maintaining approximately local interactions in two-dimensions, while preserving other desirable features of the toric code, such as braid and fusion rules.

\subsection{Encoding and Decoding Homological Product Codes}
One last set of problems concerns encoding and decoding circuits for homological product codes.
Let us start with the encoding. Consider first a single  CSS code on $M=2L+H$ qubits with $H$ logical qubits,
$L$ stabilizers of $X$ type and $L$ stabilizers of $Z$ type. It is well-known that the 
encoding  for such code can be performed by starting with $H$ qubits containing the state to be encoded, adjoining $L$ additional pairs of qubits in 
the $|0\rangle\otimes |+\rangle$ state, and then applying some unitary $M$-qubit operator $\hat{U}$ composed of CNOT gates.
Consider now a pair of such codes described by complexes $(\calC,\delta_a)$, $a=1,2$.
Let $\hat{U}_a$ be the corresponding encoding circuits composed of CNOTs.
We claim that the product code described by the complex $(C\otimes C,\partial)$ can be encoded  by applying the following steps.
\begin{enumerate}
\item Arrange $M^2$ code qubits on a two-dimensional $M\times M$ grid. 
\item Initialize some pairs of qubits in $|0\rangle \otimes |+\rangle$ state. 
\item Initialize some pairs of qubits in the EPR state $(|00\rangle + |11\rangle)/\sqrt{2}$.
\item Apply $\hat{U}_1$ to each row of the grid. 
\item Apply $\hat{U}_2$ to each column  of the grid. 
\end{enumerate}
Here the order of the last two steps does not matter as they commute with each other.
The number of qubits initialized at the steps~(2,3) is exactly $M^2-H^2$, such that there
remains $H^2$ free qubits that contain the state to be encoded. 

Indeed, Lemma~\ref{lemma:canonical} implies that $\delta_a=U_a \delta_0 U_a^{-1}$, where $\delta_0$ is the canonical
boundary operator defined in Eq.~(\ref{canonical}) and $U_a$ are some invertible $M\times M$ matrices. 
One can easily check that the canonical complex $(\calC,\delta_0)$ describes
a CSS code in which all stabilizers have weight $1$. This canonical CSS code can be encoding simply by starting with $H$ logical qubits
and adjoining $L$ pairs of qubits in $|0\rangle \otimes |+\rangle$ state.
The  code $(\calC,\delta_a)$ has stabilizer spaces
$C^Z_a=\im{\delta_a}=U_a \cdot \im{\delta_0}$ and $C^X_a=\im{\delta_a^T}=(U_a^{-1})^T \im{\delta_0^T}$.
Using Gaussian elimination, one can show that any $M\times M$ invertible matrix $U_a$ can be written as a product of $O(M^2)$ elementary matrices
with $1$s on the diagonal and a single non-zero off-diagonal entry.
Replacing each elementary matrix in the decomposition of $U_a$ by a suitable CNOT gate one obtains
a CNOT circuit $\hat{U}_a$ transforming codewords of the canonical code to codewords of the input code $(\calC,\delta_a)$.
Consider now the product code with the complex $(C\otimes C,\partial)$, $\partial=\delta_1\otimes I + I\otimes \delta_2$.
Let $U=U_1\otimes U_2$. Then $\partial =U(\delta_0\otimes I + I\otimes \delta_0)U^{-1}$. 
Define a ``canonical product code" corresponding to the complex $(\calC \otimes \calC ,\delta_0 \otimes I + I \otimes \delta_0)$.
Note that this code has single-qubit stabilizers $X$ and $Z$ on some qubits and two-qubit stabilizers $XX,ZZ$ on some pairs of qubits.
Hence the canonical product code can be encoded by starting with $H^2$ logical qubits and adjoining 
remaining $M^2-H^2$ qubits initialized in $|0\rangle$, $|+\rangle$, or $(|00\rangle+|11\rangle)/\sqrt{2}$ state. 
The encoding for the code $(C\otimes C,\partial)$ is the same as the encoding for the canonical code followed
by a CNOT circuit $\hat{U}$. Furthermore, since $U=(U_1\otimes I)(I\otimes U_2)$ the CNOT circuit
corresponding to $U$ is equivalent to applying the circuit $\hat{U}_2$ in every row of the grid
and then applying the circuit $\hat{U}_1$ in every column of the grid. This leads to the steps~(1-5) defined above. 
Note that in the worst case $\hat{U}_a$ consists of $O(M^2)$ CNOT gates. Hence the product code encoding
requires a circuit of size  $O(n^{3/2})$ and depth $O(n)$, where $n=M^2$ is the code length.
One can similarly construct encoding circuits for the $m$-fold homological product. In this case the encoding circuit has size $O(M^{(m+1)/m})$
and depth $O(n^{2/m})$. Here we assumed that $m=O(1)$. 
One interesting open question is whether the encoding for the $m$-fold product code can be implemented
in a fault-tolerant fashion, such that the overall encoding circuit is represented as
a composition of small-depth circuits and error correction operations. 
Another interesting open question arises if $(\calC,\delta_a)$   has a short preparation circuit for code states taking 
$o(M^2)$ operations that does {\it not} give us a short circuit to implement $\hat{U}_a$;
 in this case, it is not clear if there must be an similarly fast preparation circuit for the product code.

Let us now discuss the decoding.  Given a noise model, such as a random Pauli channel, and a set of syndromes measured on a corrupted codeword,
how efficiently can we determine the optimal recovery operation composed of Pauli $X$ and $Z$ 
that minimizes the probability of a logical error?  For a random code on $M$ qubits, we expect this to take time exponential in $M$ as finding optimal decodings of stabilizer codes is \#P-complete\cite{IyerPoulin}.
However, for a product of two such codes, with $n=M^2$ qubits, a decoding time exponential in $M$ would be exponential in $O(\sqrt{n})$; while this would not be polynomial, it might be practical for small enough $n$, and so it would be very desirable if a decoding algorithm with that complexity could be found for product codes.  A natural candidate is a message passing algorithm as in Ref.~\onlinecite{PoulinChung}.  Such messsage passing algorithms encounter problems for quantum LDPC codes with stabilizers of weight $O(1)$ due to degeneracy 
of the code because some low-weight errors are not uniquely determined by their syndromes.  However, 
since the homological product of two random codes has stabilizers of weight $\Theta(\sqrt{n})$, 
one might hope that message passing will converge quickly to the optimal recovery operator.  Such a message passing algorithm would take time $O(n) 2^{\sqrt{n}}$ per round, as the number of possible combinations of messages would be $2^{\sqrt{n}}$.
We leave the analysis of such an algorithm as an open problem.

\section*{Appendix A: Counting Matrix Extensions}

The purpose of this section is to prove Eqs.~(\ref{E},\ref{Esimple}).
We first prove Eq.~(\ref{Esimple}).
Let $Y$ be a rank-$R$ matrix of size $A\times B$. 
Any such matrix can be represeted as $Y=F G$ for some full-rank
matrix $F$ of size $A\times R$ and full-rank matrix $G$ of size $R\times B$.
Moreover, this representation is unique up to a transformation
$F\to FM$ and $G\to M^{-1}G$, where $M$ is an arbitrary invertible $R\times R$ matrix.
It is well-known that the number of full-rank matrices of size $a\times b$ is
$O(1)\cdot 2^{ab}$. Hence the total number of rank-$R$ matrices of size $A\times B$ is 
\[
E^{A,B,R}=O(1) \cdot 2^{(A+B)R -R^2}.
\]
This proves Eq.~(\ref{Esimple}). Let us now prove Eq.~(\ref{E}).
Recall that $E_{a,r}^{A,R}$ is the number of ways to extend a given rank-$r$ matrix $X$
of size $a\times a$ to an arbitrary rank-$R$ matrix $Y$ of size $A\times A$. 
We shall  extend $X$ to $Y$ in two steps as shown below. 
\begin{equation}
\label{XtoZtoY}
X\to Z=\left[ \begin{array}{c}
X \\ U \\ 
\end{array}
\right]
\to Y=\left[ \begin{array}{cc}
Z & V \\
\end{array}
\right].
\end{equation}
Denoting $z=\rnk{(Z)}$ we arrive at 
\begin{equation}
\label{Eproof1}
E_{a,r}^{A,R} = \sum_{z=r}^{\min{\{R,a\}}} E_{a,a,r}^{A,a,z} \cdot E_{A,a,z}^{A,A,R}.
\end{equation}
Here $E_{a,b,r}^{A,B,R}$ denotes the number of ways to extend a given rank-$r$ matrix
of size $a\times b$ to an arbitrary rank-$R$ matrix of size $A\times B$.
Let $\calM(a,b)$ be the set of all binary $a\times b$ matrices. 
Since the number of extension depends only on the rank of the original matrix,
we can compute 
 $E_{a,a,r}^{A,a,z}$  by choosing $X$ as any fixed matrix of rank $r$.
Choose  $X$ be the diagonal matrix such that $X_{i,i}=1$ for $1\le i\le r$ and $X_{i,j}=0$ otherwise.
 Then $Z$ can be written as a block matrix 
\[
Z=\left[ \begin{array}{cc} I & 0 \\ 0 & 0 \\ V & W \\ \end{array} \right], \quad V \in \calM(A-a,r), \quad W\in\calM(A-a,a-r).
\]
Here $I$ is the identity matrix in $\calM(r,r)$. The first $r$ columns of $Z$ are independent from each other and from other columns of $Z$
regardless of the choice of $V,W$. Thus $z=r+w$, where $w=\rnk{(W)}$. 
Since there are $2^{(A-a)r}$ ways to choose $V$ and $E^{A-a,a-r,z-r}$  ways to choose $W$, we get
\begin{equation}
\label{ext1}
E_{a,a,r}^{A,a,z}=2^{(A-a)r} \cdot E^{A-a,a-r,z-r} = O(1)\cdot 
 2^{(A+r)z  -a r  - z^2 }. 
\end{equation}
Repeating exactly the same arguments yields
\begin{equation}
\label{ext2}
E_{A,a,z}^{A,A,R}= 2^{(A-a)z} \cdot E^{A-z,A-a,R-z} =O(1)\cdot 2^{(2A-a)R-Az - R^2 +Rz}.
\end{equation}
Substituting Eqs.~(\ref{ext1},\ref{ext2}) into Eq.~(\ref{Eproof1}) results in
\begin{equation}
\label{ext3}
E_{a,r}^{A,R}= O(1)\cdot 2^{(2A-a)R - ar -R^2} \sum_{z=r}^{\min{\{a,R\}}}
2^{-z^2 + (r+R)z}.
\end{equation}
The function $2^{-z^2 + (r+R)z}$ has a maximum at $z=z_0=(r+R)/2$ and decays exponentially away
from the maximum. We can bound the sum over $z$ from above by extending the summation
range to all integer $z\ge 0$ and approximating the sum, up to a factor $O(1)$, by the largest term
 $2^{-z_0^2 + (r+R)z_0}=2^{(r+R)^2/4}$. This gives Eq.~(\ref{E}).

\section*{Appendix B: $GF(4)$-linear codes}
\label{sec:appB}

\subsection*{Chain Complexes from $GF(4)$-linear Codes}
In this section we  propose one possible way to extend the mapping between quantum codes and chain complexes to $GF(4)$-linear codes~\cite{Calderbank98}.
We begin by  recalling the construction of $GF(4)$-linear codes introduced in Ref.~\onlinecite{Calderbank98}.
Let $\omega$ be the multiplicative generator of $\FF_4\equiv GF(4)$ such that $\FF_4=\{0,1=\omega^3,\omega,\omega^2\}$. 
The addition in $\FF_4$ is defined by identities $1+\omega+\omega^2=0$ and $x+x=0$ for any $x\in \FF_4$.
Below we consider vectors and matrices with entries from $\FF_4$. 
A subset $C\subseteq \FF_4^n$ is called a {\em linear subspace} iff $C$ is closed under addition of vectors
and under a scalar multiplication by $\omega$. 
To describe quantum $\FF_4$-linear codes, parameterize single-qubit Pauli operators $X,Y,Z$ and the identity operator $I$ by 
elements of $\FF_4$ as
\begin{equation}
\label{Paulis}
P(0) =  I,\quad
P(\omega) = X,\quad 
P(\omega^2)  =  Z,\quad
P(1)= Y.
\end{equation}
Note that addition in $\FF_4$ corresponds to multiplication of Pauli operators, that is, $P(a)P(b)=e^{i\theta} P(a+b)$
for some phase factor $e^{i\theta}\in \{1,\pm i\}$ that depends on $a$ and $b$. Furthermore, 
\begin{equation}
\label{comm}
P(a)P(b)=(-1)^{\bar{a}b+a\bar{b}} P(b) P(a), \quad \mbox{where}\quad \bar{a}\equiv a^2.
\end{equation}
Note that $\bar{c}+c$ takes values $0$ or $1$ for any $c\in \FF_4$, so that Eq.~(\ref{comm}) is well-defined.
Given a vector $f=(f_1,\ldots,f_n)\in \FF_4^n$, let $P(f)$ be the $n$-qubit Pauli operator
that acts on the $j$-th qubit as $P(f_j)$. Then Eq.~(\ref{comm}) implies
\begin{equation}
\label{comm1}
P(f) P(g)=(-1)^{(f,g) + (g,f)} P(g)P(f),
\end{equation}
where $(f,g)\in \FF_4$ is the inner product between vectors $f,g\in \FF_4^n$ defined as
\begin{equation}
\label{inner}
(f,g)=\sum_{j=1}^n \bar{f}_j g_j.
\end{equation}
Given a linear subspace $C\subseteq \FF_4^n$, the following
three conditions are known to be equivalent~\cite{Calderbank98}: 
\begin{enumerate}
\item $(f,g)+(g,f)=0$ for any $f,g\in C$.
\item $P(f)P(g)=P(g)P(f)$ for any $f,g\in C$.
\item $(f,g)=0$ for any $f,g\in C$.
\end{enumerate}
If one of the above
conditions is satisfied, we will say 
that $C$ is {\em self-orthogonal}.
Given any self-orthogonal linear subspace $C\subseteq \FF_4^n$ one can define   a quantum stabilizer
code with a stabilizer group $G=\{P(f)\, : \, f\in C\}$. Note that $G$ has size $4^{\dim{(\calC)}}$. 
The subspace $C$ defines parity checks of the quantum code and plays the same role
as the pair of parity check spaces $C^Z,C^X$ in the case of CSS codes. 
To describe Pauli operators commuting with stabilizers
define an orthogonal subspace
\[
C^\perp=\{ f\in \FF_2^n \, :\, (f,g)=0 \quad \mbox{for all $g\in C$}\}.
\]
Then $P(f)$ commutes with all stabilizers iff $f\in C^\perp$. Logical Pauli operators 
have a form $P(f)$, where $f\in C^\perp\backslash C$.
Note that $C$ is self-orthogonal iff $C\subseteq C^\perp$. This condition plays the same role
as the orthogonality condition $C^Z\subseteq (C^X)^\perp$ in the case of CSS codes. 
As was shown in Ref.~\onlinecite{Calderbank98},  the quantum code corresponding to $C$
 has parameters $[[n,k,d]]$, where
\begin{equation}
\label{nkd}
k=n-2\dim{(\calC)} \quad \mbox{and} \quad d=\min_{f\in C^\perp\backslash C} \mathrm{wt}(f).
\end{equation}
Here $\mathrm{wt}(f)$ is the {\em weight} of $f$ defined as the number of non-zero components of $f$.
 
Given a linear operator $\delta$ mapping $\FF_4^n$ to itself, define an adjoint operator
$\delta^*$ such that $(f,\delta g)=(\delta^*f, g)$ for all $f,g\in \FF_4^n$. 
One can easily check that $\delta^*=\bar{\delta}^T$, that is,
$\delta^*_{i,j}=\bar{\delta}_{j,i}$. Note that $(\delta^*)^*=\delta$ since $x^4=x$ for any $x\in \FF_4$.
Here and below by a linear operator we always mean $\FF_4$-linear
operator. By analogy with the single sector theory for CSS codes, see Section~\ref{subs:CSS}, let us introduce
a notion of a {\em boundary operator} such that $\im{\delta}$ is a self-orthogonal linear subspace
for any boundary operator $\delta$. 
\begin{lemma}
\label{lemma:boF4}
Suppose $\delta$ is a linear operator. Then  $\im{\delta}$ is self-orthogonal iff $\delta^*\delta=0$. 
\end{lemma}
\begin{proof}
Suppose $\im{\delta}$ is self-orthogonal.  Then 
for any vectors $f,g\in \FF_4^n$ one has $(g,\delta^*\delta f)=(\delta g,\delta f)=0$. This is only possible if $\delta^*\delta=0$.
Conversely, suppose $\delta^*\delta=0$. Choose any vectors $f,g\in \im{\delta}$. Then $f=\delta(h)$ and $g=\delta(k)$
for some $h,k\in \FF_4^n$. Thus $(f,g)=(\delta h,\delta k)=(h,\delta^*\delta k)=0$, that is, $\im{\delta}$ is self-orthogonal. 
\end{proof}
The above lemma suggests that a boundary operator could be defined by a condition $\delta^*\delta=0$. 
This definition however is not quite satisfactory because it is not stable under the product of complexes. 
Indeed, suppose $\delta_1,\delta_2$ are linear operators satisfying $\delta_a^*\delta_a=0$. Define
$\partial=\delta_1\otimes I + I \otimes \delta_2$. Then $\partial^*\partial=\delta_1^*\otimes \delta_2+\delta_1\otimes \delta_2^*$
and thus generally $\partial^*\partial\ne 0$. Instead, we choose the following definition.
\begin{definition}
A linear operator $\delta$ mapping $\FF_4^n$ to itself is called a boundary operator
if it is self-adjoint, $\delta^*=\delta$, and satisfies $\delta^2=0$. 
\end{definition}
Lemma~\ref{lemma:boF4} implies that $\im{\delta}$ is a self-orthogonal linear subspace
for any boundary operator $\delta$ and $(\im{\delta})^\perp=\kr{\delta}$.
Thus any boundary operator $\delta$ on $\FF_4^n$  defines a quantum code $[[n,k,d]]$ with parameters
\[
k=\dim{(\kr{\delta})}-\dim{(\im{\delta})}\equiv H(\delta) \quad \mbox{and} \quad 
d=\min_{f\in \kr{\delta}\backslash \im{\delta}} \mathrm{wt}(f).
\]
Conversely, given 
a self-orthogonal linear subspace $C\subseteq \FF_4^n$, choose any linear basis
$a^1,\ldots,a^m\in C$ and define a linear operator
\begin{equation}
\label{bo}
\delta=\sum_{i,j=1}^m U_{i,j} \, a^i (\bar{a}^j)^T
\end{equation}
for some invertible self-adjoint matrix $U$ with $\FF_4$ entries. 
Self-orthogonality of $C$ implies $(a^i,a^j)=0$ for all $i,j$, that is, $\delta^2=0$.
Furthermore, 
\[
\delta^*=\sum_{i,j=1}^m \bar{U}_{j,i} \, a^i (\bar{a}^j)^T=\delta
\]
since $\bar{U}_{j,i}=U_{i,j}$. Finally, $\im{\delta}=C$ since $U$ is invertible. 
This shows how to represent any $\FF_4$-linear quantum code by a boundary operator.

Consider some $\FF_4$-linear quantum code $[[n,k,d]]$ described by a self-orthogonal subspace $C\subseteq \FF_4^n$.
Let 
\be
\label{delta12}
\delta_1=\sum_{i,j=1}^m U_{i,j} \, a^i (\bar{a}^j)^T \quad \mbox{and} \quad \delta_2=\sum_{i,j=1}^m V_{i,j} \, a^i (\bar{a}^j)^T
\ee
be the boundary operators constructed above such that $\im{\delta_a}=C$. 
Define $\partial=\delta_1\otimes I + I\otimes \delta_2$. One can easily check that  $\partial^2=0$ and 
$\partial^*=\partial$, that is, $\partial$ is a boundary operator on $\FF_4^n\otimes \FF_4^n$.
Using exactly the same arguments as in the proof of Lemmas~\ref{lemma:Kun},\ref{lemma:Kun1}
one can show that the $\FF_4$-linear quantum code  with the parity check space $\im{\partial}$
has parameters $[[n^2,k^2,d']]$ for some $d\le d'\le d^2$. Furthermore, if $C$ has basis vectors with weight at most $w$
then $\im{\partial}$ has basis vectors with weight at most $2w$, that is, the product code has parity checks of weight
at most $2w$.

\subsection*{Homological Product of Two $[[5,1,3]]$ Codes}
Let us  apply the product construction defined above to the $5$-qubit
code $[[5,1,3]]$ which is the simplest quantum code correcting any single-qubit  error~\cite{Bennett96,Laflamme96}.
Recall that the $[[5,1,3]]$ code has two-dimensional parity check space $C\subseteq \FF_4^5$ with 
basis vectors 
\[
a^1=(0, \omega, \omega^2, \omega^2,\omega)^T \quad \mbox{and} \quad a^2=(\omega, 0,\omega, \omega^2,\omega^2)^T.
\]
Note that $a^2$ is a cyclic shift of $a^1$. Moreover any non-zero vector of $C$ can be obtained from $a^1$ by cyclic shifts and
a scalar multiplication by $\omega$. Hence any non-zero vector of $C$ has weight $4$. 
Let $\delta_1,\delta_2$ be the boundary operators defined in Eq.~(\ref{delta12}) and $\partial=\delta_1\otimes I + I\otimes \delta_2$.
Since any column and any row of $\delta_a$ has weight $0$ or $4$ for any choice of $U$ and $V$, the product code described by $\partial$
has parity checks of weight at most $8$. One can easily check that there are
only $10$ invertible self-adjoint matrices of size $2\times 2$ with $\FF_4$ entries. We computed the distance
of the product code numerically 
for each choice of the pair $U,V$ by performing the exhaustive search over all non-trivial cycles. Note that there
are $2^{25-1}=2^{24}$ cycles to be considered. We observed that the product code has parameters $[[25,1,5]]$ 
regardless of the choice of $U,V$.
For comparison, concatenation of two $5$-qubit codes gives $[[25,1,9]]$ code
with parity checks of weight $12$. 

Another simple example of an $\FF_4$-linear code is the Steane code $[[7,1,3]]$. It has three-dimensional parity check space
$C\subseteq \FF_4^7$ with basis vectors $a^1,a^2,a^3$ defined in Eq.~(\ref{eq:A}).  Thus the product of $\FF_4$-linear
codes $[[5,1,3]]$ and $[[7,1,3]]$ is well-defined and has parameters $[[35,1,d]]$ for some $3\le d\le 9$. 
By computing the distance of the product code numerically we observed that $d\le 6$
for all
possible choices of boundary operators $\delta_1,\delta_2$ describing the codes $[[5,1,3]]$ and $[[7,1,3]]$
(the exhaustive search over non-trivial cycles was terminated as soon as the first cycle with weight at most $6$ has been
found).

These observations suggest that the single sector theory  does not perform very well when 
applied to $\FF_4$-linear codes. We leave explanation of this phenomenon for a future work. 

\vspace{10mm}

{\bf Acknowledgments --}
We would like to thank Michael Freedman and Alexey Kitaev  for  helpful discussions.
We thank Maris Ozols and Alex Vargo for valuable suggestions on numerical computation of the product code distance. 
We thank Aram Harrow for sharing with us a preliminary version of the manuscript Ref.~\onlinecite{bacon}. 
SB is supported in part by the DARPA QuEST program under contract number HR0011-09-C-0047, and
IARPA QCS program under contract number D11PC20167.


\end{document}